\documentclass[12,reqno]{amsart}
\usepackage{graphicx}

\usepackage{hyperref}
\usepackage{a4wide}

\usepackage{amssymb,amsthm,amsmath}
\usepackage{dsfont}
\usepackage{pgfplots}
\usepackage{enumitem} 
\usepackage{mathtools}
\usepackage{amsfonts}
\usepackage{empheq, cases}
\usepackage{float}
\usepackage{array,booktabs}

\usepackage{colortbl}
\usepackage{xcolor}
\usepackage{tablefootnote}
\usepackage[justification=centering]{caption}

\numberwithin{equation}{section}
\numberwithin{figure}{section}

\usetikzlibrary{arrows,decorations.markings}

 \renewcommand{\i}{{\rm i}}

 \newcommand{\R}{\mathbb{R}}
 \newcommand{\C}{\mathbb{C}}
 
 \newcommand\given[1][]{\:#1\vert\:}

\newtheorem{assumption}{Assumption}
\newtheorem{theorem}{Theorem}

\newtheorem{lemma}{Lemma}

\newtheorem{prop}{Proposition}
\date{}
\title[Valuation of life insurance policies]{On the valuation of life insurance policies for dependent coupled lives}

\author{Kira Henshaw}
\address{University of Liverpool}
\email{kirahenshaw@gmail.com}

\author{Cedric H. A. Koffi}
\address{Institute for Financial and Actuarial Mathematics, Department of Mathematical Sciences, University of Liverpool, L69 7ZL, United Kingdom}

\author{Olivier Menoukeu Pamen}
\address{Institute for Financial and Actuarial Mathematics, Department of Mathematical Sciences, University of Liverpool, L69 7ZL, United Kingdom}
\email{menoukeu@liverpool.ac.uk}

\author{Raghid Zeineddine}
\address{Institute for Financial and Actuarial Mathematics, Department of Mathematical Sciences, University of Liverpool, L69 7ZL, United Kingdom}
\email{raghid.zeineddine@liverpool.ac.uk}

\thanks{O. Menoukeu Pamen acknowledges the funding provided by the Alexander von Humboldt Foundation, under the programme financed by the German Federal Ministry of Education and Research entitled German Research Chair No 01DG15010}

\subjclass{Primary 91G05, 91G20; Secondary 65D30, 60G51}

\keywords{Life insurance; L\'evy processes; hybrid model;  broken-heart syndrome; bereavement effect; joint mortality}

\begin{document}

\maketitle
	\begin{abstract}
In this paper, we investigate a complex variation of the standard joint life annuity policy by introducing three distinct contingent benefits for the surviving member(s) of a couple, along with a contingent benefit for their beneficiaries if both members pass away. Our objective is to price this innovative insurance policy and analyse its sensitivity to key model parameters, particularly those related to the joint mortality framework. We employ the $QP$-rule (described in Section \ref{secgenset}), which combines the real-world probability measure $P$ for mortality risk with risk-neutral valuation under $Q$ for financial market risks. The model enables explicit pricing expressions, computed using efficient numerical methods. Our results highlight the interdependent risks faced by couples, such as broken-heart syndrome, providing valuable insights for insurers and policyholders regarding the pricing influences of these factors.

		\end{abstract}

\section{Introduction}
A standard life insurance policy is a contract between an individual and an insurance company that provides periodic payouts, or annuities, to the insured throughout their lifetime. A well-known variation is the last-survivor annuity or joint life annuity, which continues annuity payments until the last member of a couple dies. For more details on these products and their valuation, see \cite{Brown, Frees, Luciano, Luciano2, Sari} and others. Joint life tables differ from individual life tables, making it important to model the joint survivorship of couples. Cohabiting couples' lives are often interdependent, especially with short-term mortality impacts like broken-heart syndrome \cite{Fitz, Ward}. This effect, along with long-term impacts such as life circumstances bereavement \cite{Hougaard}, has been modeled in continuous time by \cite{jevtic2017joint}, a model we adopt in this paper.

In our paper, we introduce a complex variation of the standard joint life annuity policy by offering three distinct contingent benefits to the surviving member(s) of a couple, along with a contingent benefit to their beneficiaries in the event that both members pass away. To the best of our knowledge, this represents a novel type of insurance policy specifically designed for couples.

Our primary aim is to determine the price of this innovative insurance policy and explore the sensitivity of its pricing to various model parameters, particularly those within the joint mortality framework. The joint mortality model plays a crucial role in capturing the interdependent risks between the couple, such as the broken-heart syndrome and other bereavement effects. Understanding how these parameters impact the policy price helps inform both the insurance provider and the policyholder of the associated risks.

Below, we provide a detailed description of the policy, explaining its contingent benefits and outlining the structure that distinguishes it from other life insurance products.

Let $T^*>0$ be a finite time horizon, and let $(S_t)_{0\leq t \leq T^*}$ be the price of a generic equity index, the dynamic of this price is described later in \eqref{Stock price}. The insurance policy we examine provides various benefits to policyholders based on the notional $I$, in exchange for a premium of the policy that must be paid when entering the contract. The notional $I$ is invested in the equity index. Our first aim is to determine this premium. At maturity $T\in (0,T^*]$, the contract offers a \textsl{Guaranteed Minimum Accumulation Benefit (GMAB)},  which allows for growth opportunities based on the equity index’s performance. Additionally, the contract permits early termination, referred to as surrender, in which case a \textsl{Surrender Benefit (SB)} is paid by the insurance company. Moreover, the contract includes a  \textsl{Death Benefit (DB)}, payable upon the death of either member of the couple. These benefits are standard guarantees in variable annuity contracts, as seen in sources like \cite{Baci, VA1, VA2, Escobar}. Let $\tau_{x_1}$ and $\tau_{x_2}$ denote the remaining lifetimes of the two members of the couple aged $x_1$ and $x_2$, respectively at time 0, and let $\tau^s$ denote the random time of  surrender. Mathematical models for these random times are given in Section \ref{secmomo}. Below, we provide a detailed description of each benefit.

 At maturity $T$, if the surrender option has not been exercised and at least one member of the couple is still alive, then the GMAB pays 
\[ \max( I S_T, I e^{\delta T}), \]
where  $\delta>0$ is a guaranteed rate. The GMAB provides opportunities for growth based on the performance of the equity index, while also offering protection through the guaranteed benefit $I e^{\delta T}$. 
Using the random times introduced earlier, the guaranteed minimum accumulation benefit is expressed as:
\begin{equation}
\text{GMAB}(T)=\max(IS_T,Ie^{\delta T})\mathds{1}_{\{\tau^s>T\}}\mathds{1}_{\{\tau_{x_1}>T\}\cup\{\tau_{x_2}>T\}}. \label{formula GMAB}
\end{equation}
In other words, at time $T$, if there is no early surrender and at least one member of the couple is alive, they receive the GMAB payoff. If both members die before maturity $T$, no GMAB is paid.

Concerning the surrender benefit, let ${\bf t}:=(t_0,t_1, \ldots, t_K)^\top$ with $0=t_0<t_1< \ldots <t_K< T$. Early surrender is possible at any time point $t_i\in {\bf t}$ with $i = 1, \ldots,K-1$. If at least one member of the couple is alive at time $t_i\in {\bf t}$, and the contract is cancelled, the policyholder is entitled to a refund equal to the current value of their share underlying index ($IS_{t_i}$), reduced by a surrender penalty.  The surrender benefit pays  at time $t_i$:
\[    IS_{t_i}- IS_{t_i}\bar P(t_i), \]
where $\bar P:[0,T] \to [0,1)$ is a decreasing penalty function with $\bar P(T)=0$, meaning that early surrender incurs a higher penalty. Let $\widetilde P:=1- \bar P$, so $ \widetilde P:[0,T] \to (0,1]$ is an increasing function with $\widetilde P(T)=1$. Thus the surrender benefit simplifies to: 
\[    IS_{t_i}\widetilde P(t_i). \]
 Using the random times $\tau_{x_1}, \tau_{x_2}$, and $\tau^s$, the surrender benefit at  time $t_i\in {\bf t}$ is given by

\begin{equation}
\text{SB}(t_i)=IS_{t_i}\widetilde P(t_i)\mathds{1}_{\{\tau^s=t_i\}} \mathds{1}_{\{\tau_{x_1}>t_i\}\cup\{\tau_{x_2}>t_i\}}. \label{formula SB}
\end{equation}
In other words, at time $t_i$,  if only the member of the couple aged $x_1$ at time 0 is still alive and they decide to terminate the contract, they receive the Surrender Benefit (SB).  The same applies if only the member aged $x_2$ is alive and decides to terminate the contract. If both members are alive, the decision to terminate the contract must be mutual, as neither can cancel without the other’s agreement.

Regarding the Death Benefit (DB), let $\bar{\bf t}:=(\bar t_0,\bar{t}_1, \ldots, \bar{t}_N)^\top$ with $0=\bar t_0 < \bar{t}_1< \ldots < \bar{t}_N=T$, where the mortality is monitored at these time points. If no surrender occurs before $\bar{t_i}$, and only one member dies in $[\bar t_{i-1}, \bar t_i[$, the surviving member receives
\[ \max(IS_{\bar{t}_i}, Ie^{\delta \bar{t}_i}).\]
If both members die in $[\bar t_{i-1}, \bar t_i[$, the beneficiaries receive
\[\alpha  \max(IS_{\bar{t}_i}, Ie^{\delta \bar{t}_i}),\] 
where $1<\alpha <2$, indicating a higher payout when both members die in the same interval.

Although the time scales ${\bf t}$ and $\bar {\bf t}$ are arbitrary, it is natural to assume that ${\bf t}\subset \bar {\bf t}$, allowing surrender at annual intervals and monitoring of death more frequently, such as quarterly. Using the random times $\tau_{x_1}, \tau_{x_2}$, and $\tau^s$, the Death Benefit at  time $\bar t_i\in {\bf \bar t}$ is given by:
\begin{equation}
\text{DB}(\bar{t}_i)=\max(IS_{\bar{t}_i},Ie^{\delta \bar{t}_i})\mathds{1}_{\{\tau^s\geq \bar{t}_i\}}[\mathds{1}_{\{\bar{t}_{i-1}\leq\tau_{x_1}<\bar{t}_i\}}+\mathds{1}_{\{\bar{t}_{i-1}\leq\tau_{x_2}<\bar{t}_i\}} +(\alpha-2)\mathds{1}_{\{\bar{t}_{i-1}\leq\tau_{x_1}<\bar{t}_i\}}\mathds{1}_{\{\bar{t}_{i-1}\leq\tau_{x_2}<\bar{t}_i\}}], \label{formula DB}
\end{equation}
where $1< \alpha <2$.

It is clear that the benefits of our insurance policy are linked to the financial market, as is the case with the surrender risk, based on the definition of surrender intensity provided in Section \ref{secmomo}. It is commonly agreed that market valuation should be conducted under a risk-neutral probability measure $Q$. Simultaneously, the insurer faces mortality risk from the policyholders, which must be evaluated under the real-world probability measure $P$. To address this, we use the $QP$-rule, presented in Section \ref{secgenset}, which combines market valuation under $Q$ with mortality risk evaluation under $P$. 

Our work contributes to the existing literature in several key ways. First, we price a life insurance policy designed for a couple or spousal pair, modeled as two cohabiting adults, using the flexible continuous-time joint mortality model of \cite{jevtic2017joint}. Second, the policy provides three contingent benefits, all linked to the financial market, to the couple and their beneficiaries. A notable feature of the death benefit is that it covers not only the surviving spouse but also the beneficiaries in the event of both partners' deaths. Third, we employ the $QP$-rule (described in Section \ref{secgenset}) to combine the real-world probability measure $P$ for mortality risk with the risk-neutral valuation under $Q$ for financial market risks. Lastly, the paper presents a sensitivity analysis of the policy price, focusing particularly on the mortality model parameters and demonstrating the impact of the broken-heart syndrome on pricing.

The paper is organised as follows. In Section \ref{secgenset} we provide some preliminary results about  L\'evy processes, we introduce the model for the financial market and we introduce the $QP$-rule. In Section \ref{secmomo} we introduce the mortality and surrender models. Section \ref{secpricing} focuses on the pricing framework of the policy and states the main theorems; then the proofs of these theorems are given in Section 5. The numerical implementation and the sensitivity analysis are  presented in Section 6. Finally, some technical representations are given in the appendix.

\section{The general setup}\label{secgenset}
Let $T^*>0$ be a finite time horizon, let  $(\Omega, \mathcal{H})$ be a measurable space and let $P$ be the real-world probability measure defined on $\mathcal{H}$. Let $\mathbb{F}=(\mathcal F_t)_{t\in [0,T^*]}$ be a filtration satisfying the usual conditions, and let $\mathcal{F}:=\mathcal{F}_{T^*}$. For any $t\in [0,T^*]$, $\mathcal{F}_t$ contains only the available information about the financial market up to time $t$; we have $\mathcal{F} \subset \mathcal{H}$, in fact $\mathcal{H}$ contains more information than $\mathcal{F}$ such as national mortality rates, life-tables, health situations of the clients of the insurance company offering the insurance contract that we have in this paper, etc. Let $Q$ be a probability measure  defined on $\mathcal{F}$, we assume that $Q$ is equivalent to the restriction of $P$ to $\mathcal{F}$. 

\subsection{Preliminaries}
In this section, we present some preliminary results on Lévy processes, which drive the dynamics of the equity and interest rate markets.

  Consider the stochastic basis $(\Omega, \mathcal F, \mathbb{F}, Q)$, 
let $L^1$ and $L^2$ be two independent time-inhomogeneous L\'evy processes, i.e., continuous-time processes with independent increments, and the following characteristic function:
\begin{align}\label{eq:FourierTransformL}
  \mathbb{E}_Q\Big[e^{iuL^j_t}\Big]= \exp \Big\{\int_0^t\Big(iub_s^j -\frac12c_s^ju^2 
  + \int_{\mathbb{R}}\big(e^{iux}-1-iux {\bf 1}_{|x|\leq 1}\big)F^j_s(\mathrm{d}x)\Big)\mathrm{d}s\Big\}
\end{align}
for $j=1,2$; where the functions $(b_s^j, c_s^j, F_s^j)_{s\in [0, T^*]},\, j=1,2$ satisfy
the integrability condition
\[
\int_0^{T^*}\Big(|b_s^j| + c_s^j+ \int_{\mathbb{R}}(\min{\{|x|^2, 1\}})F_s^j(\mathrm{d}x)\Big)\mathrm{d}s < \infty.
\]

We require the existence of exponential moments of a certain order. Therefore, we introduce the following assumption.

\begin{assumption}[Exponential moments]
\label{ass:EM}
For $j=1,2$, there exist positive constants $M_j$ and $\epsilon_j$, such that for each $u \in [-(1+\epsilon_j)M_j, (1+\epsilon_j)M_j]$ 
\[
\int_0^{T^*}\int_{\{|x|> 1\}}e^{ux}\, F^j_s(\mathrm{d}x)\mathrm{d}s < \infty .
\]
\end{assumption}

It is worth pointing out that standard L\'evy processes, such as hyperbolic, normal Inverse Gaussian, Variance Gamma and CGMY processes, satisfy the above condition.
The cumulant function  of $L^j, j=1,2$ is defined as follows
\[
\theta^j_s(z) = b_s^jz + \frac12c_s^jz^2 + \int_{\mathbb{R}}(e^{zx} -1 -zx)F_s^j(\mathrm{d}x),
\]
for any $z\in \mathbb{C}$ such that $Re(z) \in [-(1+\epsilon_j)M_j, (1+\epsilon_j)M_j]$. 
Consequently for any $z\in \mathbb{C}$ verifying the previous condition, we have 
\[
\mathbb{E}_Q[\exp(zL^j_t)]< \infty \text{ and }	\mathbb{E}_Q\big[\exp(zL^j_t)\big] = \exp\Big(\int_0^t \theta^j_s(z) \mathrm{d}s \Big).
\]
We have the following important result, which will very useful in the sequel (see \cite{EberleinRai} for more details).  Let $f:\mathbb{R}^+ \to \mathbb{C}$ be a left-continuous function with $|Re(f)| \leq M_j$, then
\begin{equation}\label{cumulant}
	\mathbb{E}_Q\Big[\exp\Big(\int_0^tf(s)\mathrm{d}L^j_s\Big)\Big] = \exp\Big(\int_0^t \theta^j_s(f(s)) \mathrm{d}s \Big),
\end{equation}
where the integrals are defined component-wise for the real and imaginary part. 

\subsection{Financial and Bond  markets}

In \cite{Eberlein-R}, the authors developed two versions of a hybrid model for equity and interest rate markets. Their approach is termed "hybrid" because it accounts for the stochastic dependence between the two markets, a well-known phenomenon. In their models, the dynamics of both the equity and interest rate markets are driven by time-inhomogeneous Lévy processes. The first version of their model was adopted in \cite{VA1, VA2}; in this paper, we adopt their second version, called the Hybrid Lévy Forward Rate Model. This model is described as follows: Consider the stochastic basis $(\Omega, \mathcal F, \mathbb{F}, Q)$ introduced in the previous subsection, let $L^1=(L^1_t)_{t\in [0,T^*]}$ and $L^2=(L^2_t)_{t\in [0,T^*]}$ be two independent time-inhomogeneous L\'evy processes. The instantaneous forward rates $(f(t,T)_{0 \le t \le T \le T^*})$ are given by  
\begin{equation}
	f(t,T) = f(0,T) + \int_0^t \alpha(s,T)\mathrm{d}s - \int_0^t \sigma_1(s,T)\mathrm{d}L^1_s +\int_0^t\beta(s,T)\mathrm{d}L^2_s, \quad 0 \le t \le T \le T^*, \label{forward rate}
\end{equation}
where $f(0,T)$ is  a deterministic and bounded function. 
The drift functions $\alpha(\cdot,T)$ and $\beta(\cdot,T)$ as well as the volatility  $\sigma_1(\cdot,T)$  could be chosen to satisfy the usual conditions of measurability and boundedness (see \cite{EJR}, equation (2.5)). But, for the sake of simplicity in the calibration of the model, the authors considered the following standard choice for $\sigma_1(\cdot, T)$ and $\beta(\cdot, T)$ 
\begin{equation}\label{sigma un}
\sigma_1(s, T):=	\begin{cases}
	                 a\exp(-a(T-s)), \, &\text{ if } s\leq T\\
	                 0, \,  &\text{ if } s>T
                	\end{cases}
\end{equation}
where $a$ is a positive constant, and 
\begin{equation}\label{beta un}
	\beta(s, T):=	\begin{cases}
		b\exp(-b(T-s)), \,  &\text{ if } s\leq T\\
		0, \,  &\text{ if } s>T
	\end{cases}
\end{equation}
where $b\neq 0$.
The price  of a zero coupon bond at time $t$ with maturity $T\ge t$ is 
\[
B(t,T)= \exp\Big(-\int_t^T f(t,s)\mathrm{d}s\Big).
\]
 We denote 
\begin{eqnarray*}
A(u,T)&:=& \int_u^T \alpha(u,s)\mathrm{d}s,  \ \ \Sigma_1(u,T):= \int_u^T \sigma_1(u,s)\mathrm{d}s,\\
\Sigma_2(u,T)&:= & \int_u^T \beta(u,s)\mathrm{d}s.
\end{eqnarray*} 
The short rate is defined by $r(t):=f(t,t)$ and applying Lemma \ref{lemm1a}  (see Appendix \ref{app:prelim}), we have the following formula for $B(t,T)$
\[
B(t,T)=B(0,T)\exp\Big(\int_0^{t}r(s)\mathrm{d}s  - \int_0^{t}A(u,T)\mathrm{d}u + \int_0^{t}\Sigma_1(u,T)\mathrm{d}L^1_u - \int_0^{t}\Sigma_2(u,T)\mathrm{d}L^2_u\Big).
\]
 We further assume that 
\begin{eqnarray}
\Sigma_1(u,T)&\leq & M_1, \label{AssM1}\\
|\Sigma_2(u,T)| &\leq & \frac{M_2}{3}, \label{AssM22}
\end{eqnarray}
where $M_1$ and $M_2$ are the constants given in Assumption \ref{ass:EM}. This implies that the exponential of each stochastic integral has a finite expectation under $Q$.

For the market to be arbitrage free, we require that $(B_t^{-1}B(t,T))_{0 \le t \le T}$, where $B_t = \exp\big(\int_0^tr(s)\mathrm{d}s\big)$, is a martingale. Using the independence of $L^1$ and $L^2$ and equation \eqref{cumulant},  we deduce that   $Q$  is  a martingale measure if
\begin{equation}\label{dc}
   A(u,T)= \theta^1_u(\Sigma_1(u,T)) +\theta^2_u(-\Sigma_2(u,T)), \: \: \: 0 \le u \le T.
\end{equation}
Therefore, we assume in the following that \eqref{dc} holds.

The price of a generic equity index is given by 
\begin{equation}
S_t=S_0\exp\Big(\int_0^t r(s)ds+\int_0^t \sigma_2(s)\mathrm{d}L^2_s -\omega(t)\Big), \label{Stock price}
\end{equation}  
where the volatility  $\sigma_2(\cdot)$ is assumed to satisfy the usual conditions of measurability and boundedness. We assume, without loss of generality, that $S_0=1$. We assume that  
\begin{equation}
\sigma_2(s)\leq \frac{M_2}{3}, \label{AssM2}
\end{equation}
where $M_2$ is the constant given in Assumption \ref{ass:EM}; this implies that the exponential of the stochastic integral has a finite expectation under $Q$. The drift term $\omega(t) $ in \eqref{Stock price} is chosen such that the discounted price $(B_t^{-1}S_t)_{t \in [0,T^*]}$ is a $Q$-martingale. The formula given in \eqref{cumulant} implies that
  \begin{align}\label{dcII}
  \omega(t)= \int_0^t \theta^2_s(\sigma_2(s)) \mathrm{d}s.
  \end{align}

So we deduce from the above explanation that the probability measure $Q$ is a risk-neutral martingale measure. The risk-neutral martingale measure is not unique because a financial market built on (exponential) L\'evy processes is in general incomplete. The particular risk-neutral martingale measure $Q$ that we consider in the following is determined via calibration using market quotes for liquidly traded interest rate and equity derivatives. For a detailed description of this calibration procedure, we refer the reader to \cite{Eberlein-R}.

\subsection{The $QP$-rule}

The $QP$-rule is a key  tool in the  pricing framework of our insurance policy, as explained below. This rule combines the risk-neutral probability $Q$ and the real-world probability $P$ as described in the following proposition. It states that there exists a probability measure, denoted by $Q\odot P$ on $(\Omega, \mathcal{H})$ which extends the risk-neutral probability $Q$ to $(\Omega,\mathcal{H})$ in a $P$-reasonable way. For the history of this tool, the proof of this proposition, and the rich literature and results on this subject, we refer the reader to \cite{Artzner}. 

\begin{prop}\label{QP}
There exist a unique probability measure, denoted $Q\odot  P$, on $(\Omega, \mathcal{H})$, such that $Q\odot P=Q$ on $\mathcal{F}$ and for all $A\in \mathcal{H}$ it holds that $Q\odot P(A|\mathcal{F})=P(A|\mathcal{F})$. This measure $Q\odot P$ satisfies for any random variable $X\geq 0$:
\begin{itemize}
\item[(a)] for any $\sigma$-algebra $\mathcal{G}\subset \mathcal{F}$,
\begin{equation}
\mathbb{E}_{Q\odot P}[X|\mathcal{G}]=\mathbb{E}_Q[\mathbb{E}_P[X|\mathcal{F}]|\mathcal{G}], \label{QP1}
\end{equation}

\item[(b)] for any $\sigma$-algebra $\mathcal{G}$ satisfying $\mathcal{F}\subset \mathcal{G}\subset \mathcal{H}$,
\begin{equation}
\mathbb{E}_{Q\odot P}[X|\mathcal{G}]=\mathbb{E}_P[X|\mathcal{G}]. \label{QP2}
\end{equation}
\end{itemize}
\end{prop} 

We deduce from Proposition \ref{QP} that $Q\odot P$ is market-consistent because it coincides with the risk-neutral probability $Q$ on the information of the financial market, as given by $\mathcal{F}$. 

Our insurance policy is linked to  the financial market through the benefits of the policy, interest rate and surrender risk. Consequently, the insurer is exposed to financial risk. In addition, the insurer faces mortality risk associated with policyholders, which must be evaluated under the real-world probability $P$. As demonstrated in Section \ref{secpricing}, the use of the probability $Q\odot P$  in pricing the insurance policy allows for the evaluation of the mortality risk under the real-world probability measure $P$, while providing a standard risk-neutral valuation under $Q$ for the  other risks and payoffs related to the financial market.

\section{Mortality and surrender  models}\label{secmomo}

\subsection{Mortality model}
We adopt the joint mortality model from \cite{jevtic2017joint} to represent the mortality intensities of a couple or spousal pair, imagined  as two cohabiting adults who are the policyholders of the insurance contract. We suppose that the probability space $(\Omega, \mathcal H,  P)$, introduced in Section \ref{secgenset} is large enough to carry a two-dimensional Brownian motion $W=(W^1,W^2)$ independent of $\mathcal{F}$ (representing the information about the financial market), along with two unit exponentially distributed random variables $E_1$ and $E_2$, and a standard uniformly distributed random variable $U$. Within this space, we assume that $\{\mathcal{F},W,E_1,E_2,U\}$ forms an independent  collection.  Allowing for the finite time horizon  $T^*$ to be sufficiently long, the filtration generated by $W$ is denoted $(\mathcal{K}_t)_{t\in [0,T^*]}$, and the filtration generated by $(L^1, L^2)$ is denoted $(\mathcal{F}_t^{L^1,L^2})_{t\in [0,T^*]}$ (note that $\mathcal{F}_{T^*}^{L^1,L^2} \subset \mathcal{F}$).

 The model of \cite{jevtic2017joint} that we adopt here is an adaptation of the reduced-form approach commonly used in credit risk modeling, where default is modeled as a stopping time that occurs unexpectedly. The interested reader is referred to \cite{Lando} for more details about this framework. Let $\tau$ be  the first jump-time of a Cox process (also known as  doubly stochastic Poisson process) with intensity process $\lambda(t)$, representing the remaining lifetime of an individual. Then, the conditional probability that the stopping time $\tau$ exceeds time $t \geq 0$, under the probability measure $P$, satisfies the relation:
\begin{equation} \label{tau}
P(\tau > t \given[]\mathcal{K}_t) = \exp\Big({-\int_0^t\lambda(u)\mathrm{d}u}\Big).
\end{equation}
 This formulation captures the stochastic nature of mortality in the model, where the intensity $\lambda(t)$ evolves over time.

Consider two coupled lives aged $x_1$ and $x_2$ at time 0, with future lifetimes $\tau_{x_1}$ and $\tau_{x_2}$, respectively. For simplicity, we refer to the partner aged $x_1$ as $(x_1)$ and the partner aged $x_2$ as $(x_2)$.  Prior to the first death of a member of the couple, the instantaneous forces of mortality at time $t$ of $(x_1)$ and $(x_2)$ are denoted $\lambda_{x_1}(t)$ and $\lambda_{x_2}(t)$, respectively.

Following \cite{jevtic2017joint}, we assume that  $\lambda_{x_1}$ and $\lambda_{x_2}$ are Ornstein-Uhlenbeck diffusions\footnote{As noted in \cite{jevtic2017joint}, page 93, and based on the parameters in the numerical investigations, the probabilities of negative intensities over a 100-year period are exteremely small-no higher than one hundredth of a percent, and in most cases, considerably smaller (see \cite{Luciano2008} (page 10)).} 
\begin{equation}
\mathrm{d}\lambda_{x_i}(t)=\mu_i \lambda_{x_i}(t)\mathrm{d} t +\sigma_i\mathrm{d}W^i_t, \quad i\in\{1,2\}, \quad t\geq 0, \label{mortality intensities}
\end{equation}
where $W^1$ and $W^2$ are the two independent standard Brownian motions introduced above, and the parameters $\lambda_{x_i}(0)$, $\mu_i$, and $\sigma_i$ are positive constants. 

 The spouse whose death occurs first is identified as the deceased partner and denoted by $p \in \{(x_1),(x_2)\}$. Equivalently, the~spouse who survives the first death is denoted by $q \in \{(x_1),(x_2)\}$ and is referred to as the bereaved partner. The remaining lifetime of spouse $p$, denoted by $\tau_p$ is modelled as the first jump-time of a Cox process with intensity process $\lambda_{x_1}(t) + \lambda_{x_2}(t) $. So $\tau_p$ is given by
\begin{equation}
\tau_p := \inf\Big\{t \geq 0 \given[\Big] \int_0^t \lambda_{x_1}(u) + \lambda_{x_2}(u) \mathrm{d}u \geq E_1\Big\}.
\label{tau p}
\end{equation}
The standard uniformly distributed random variable $U$ is used to identify which spouse dies first by comparing it to a function of the mortality intensities at the time of the first death. Recalling that $p$ is the label assigned to the partner who dies first, it holds that
\begin{equation*}
\{(x_1)=p\} = \{\tau_{x_1} = \tau_p\} = \Big\{U \leq \frac{\lambda_{x_1}(\tau_p)}{\lambda_{x_1}(\tau_p) + \lambda_{x_2}(\tau_p)}\Big\},
\end{equation*}
\begin{equation*}
\{(x_2)=p\} = \{\tau_{x_2} = \tau_p\} = \Big\{U > \frac{\lambda_{x_1}(\tau_p)}{\lambda_{x_1}(\tau_p) + \lambda_{x_2}(\tau_p)}\Big\}.
\end{equation*}
The above equations indicate that indicates that spouse $(x_1)$ (respectively $(x_2)$) is identified as the deceased (respectively surviving) partner if $U$ falls within a specified range based on the mortality intensities.

The remaining lifetime of spouse $q$ is denoted $\tau_q$ and is defined analogously manner, as shown later in \eqref{tau q}. Since the loss of a spouse impacts the mortality of the surviving spouse, $\tilde{\lambda}_q(t)$ is introduced for $t \geq \tau_p$ as the mortality intensity of the bereaved partner after the first death. This force of mortality, $\tilde{\lambda}_q(t)$, is an adjustment of the original process $\lambda_q(t)$, and the relationship between the two rates reflects the influence of  bereavement on the health and remaining lifetime of the surviving spouse. 
This bereavement effect is therefore defined by the modification of the mortality intensity for the bereaved partner as follows
\begin{equation*}
r_q(t):=\tilde{\lambda}_q(t) - \lambda_q(t),
\end{equation*}
Note that the modified process $\tilde{\lambda}_q(t)$ is inclusive of a~structural break at $\tau_p$ representing the instant effect on the bereaved spouse's mortality. The instantaneous rise in mortality at the time of first death $\tau_p$ is modelled as a linear combination of the mortality intensities of both spouses just before the death, denoted by $\tau_p^-$. This rise is given by:
\begin{equation}
r_q(\tau_p) := \tilde{\lambda}_q(\tau_p) - \lambda_q(\tau_p^-) = \delta_q + \epsilon_q\lambda_q(\tau_p^-) + \zeta_q\lambda_p(\tau_p^-),
\label{BHSGhana:Equation3.9}
\end{equation}
where coefficients $\delta^q, \epsilon^q$ and $\zeta^q$ are assumed to be non-negative. Intuitively, this mortality jump represents the immediate effect of bereavement, reflecting the short-term dependency between the spouses, often referred to as broken-heart syndrome. The modification of $\lambda_q$ accounts for the longer-term adjustments in the surviving spouse's mortality intensity due to changes in their life circumstances following the death. By including both spouses' mortality intensities in the estimation of the bereavement jump, as seen in equation \eqref{BHSGhana:Equation3.9}, the model incorporates unobserved shared frailties. \\
This approach explains the causal relationship between the remaining lifetimes of the couple and the contagion effect caused by the loss of a spouse. As in \cite{jevtic2017joint}, we define $r_q(t)$ as a deterministic function, with its dynamics given by
\begin{equation*}
\mathrm{d}r_q(t) = - \kappa_qr_q(t)\mathrm{d}t, \\ \text{ for } t> \tau_p, \ \ \text{with} \ \ r_q({\tau_p}) = \epsilon_q\lambda_q({\tau_p^-}),
\end{equation*}
The coefficients $\delta_q$ and $\zeta_q$ in \eqref{BHSGhana:Equation3.9} are fixed at zero for computational simplicity; however, selection of positive values for $\delta_q$ and $\zeta_q$ allows to incorporate the mortality intensities of both lives, prior to the first death, in the initial value of the bereavement effect at time $\tau_p$. Therefore, we deduce that for $t\geq \tau_p$
\begin{equation}
 r_q(t)=\epsilon_q\lambda_q(\tau_p^-)e^{-\kappa_q(t-\tau_p)}.\label{rq}
\end{equation}
Applying an approach similar to the first time of death $\tau_p$, the~second time of death $\tau_q$ is described as:
\begin{eqnarray}
\tau_q := \inf\Big\{t > \tau_p \given[\Big] \int_{\tau_p}^t \tilde{\lambda}_q(u) \mathrm{d} u \geq E_2\Big\}.
\label{tau q}
\end{eqnarray}

 The joint probability density function of $(\tau_{x_1},\tau_{x_2})$ is given in Theorem \ref{joint density} below (representing the first part of \cite[Theorem 2.1.]{jevtic2017joint}), in which expectations are taken under $P$.

\begin{theorem}[\cite{jevtic2017joint}, Theorem 2.1.]
\label{joint density}
The joint probability density function $\rho(t_{1}, t_{2})$ of the remaining lifetimes of two coupled lives $(\tau_{x_1}, \tau_{x_2})$ is {given by the reduced form~expression}
{\begin{subequations}
\begin{empheq}[left={\rho(t_{1}, t_{2})=\empheqlbrace}]{align}
&{E}_P\Big[\lambda_{x_1}(t_{1}) e^{{-\int_{0}^{t_{1}}{\lambda_{x_1}(u) + \lambda_{x_2}(u)}}\mathrm{d}u} {E}_P\Big[ \tilde{\lambda}_{x_2}({t_{2}}) e^{-\int_{t_{1}}^{t_{2}}{\tilde{\lambda}_{x_2}(u)\mathrm{d}u}} \ | \ \mathcal{K}_{T^*} \Big] \Big], \hspace{0.7cm} \textit{$t_{1} < t_{2}$}\notag \\
&{E}_P\Big[ \lambda_{x_2}({t_{2}}) e^{{-\int_{0}^{t_{2}}{\lambda_{x_1}(u) + \lambda_{x_2}(u)}}\mathrm{d}u} {E}_P\Big[ \tilde{\lambda}_{x_1}({t_{1}}) e^{-\int_{t_{2}}^{t_{1}}\tilde{\lambda}_{x_1}(u)\mathrm{d}u} \ | \ \mathcal{K}_{T^*} \Big]\Big], \hspace{0.7cm} \textit{$t_{1} > t_{2}$}. \notag
\end{empheq}
\end{subequations}}

\end{theorem}

For $i\in\{1,2\}$, let $\phi^i:=(\lambda_{x_i}(0), \mu_i,\sigma_i,\epsilon_{x_i},\kappa_{x_i})$, where $\lambda_{x_i}(0)$, $\mu_i$, and $\sigma_i$ are given in \eqref{mortality intensities}, and $\epsilon_{x_i}$ and $\kappa_{x_i}$ are given in \eqref{rq}. The joint density function $\rho(t_1,t_2)$ has been computed in \cite[Appendix]{jevtic2017joint}. Indeed, for $t_1<t_2$,
\begin{align}
&\rho(t_{1},t_{2})\notag\\
=&\Big(-\frac{{\sigma_1}^2}{2{\mu_1}^2}(e^{\mu_1t_1}-1)^2+e^{\mu_1t_1}\lambda_{x_1}(0)\Big)\times\exp\Big(-\frac{{\sigma_1}^2}{2{\mu_1}^2}\Big(-t_1+\frac{2}{\mu_1}(e^{\mu_1t_1}-1)-\frac{1}{2\mu_1}(e^{2\mu_1t_1}-1)\Big) \nonumber \\
&-\frac{1}{\mu_1}(e^{\mu_1t_1}-1)\lambda_{x_1}(0)\Big) \times \Big(-\frac{\sigma_2^2}{2\mu_2^2}(e^{\mu_2(t_2-t_1)}-1)^2+\Big(e^{\mu_2(t_2-t_1)}+\epsilon_{x_2}e^{-\kappa_{x_2}(t_2-t_1)}\Big) \nonumber \\
&\times \Big(\frac{\sigma_2^2}{2\mu_2^2}\Big(2(e^{\mu_2t_1}-1)-\Big(e^{\mu_2(t_2-t_1)}+\frac{\mu_2\epsilon_{x_2}}{\kappa_{x_2}}(1-e^{-\kappa_{x_2}(t_2-t_1)})\Big)(e^{2\mu_2t_1}-1)\Big)+e^{\mu_2t_1}\lambda_{x_2}(0)\Big) \nonumber \\
&\times\exp\Big(-\frac{{\sigma_2}^2}{2{\mu_2}^2}\Big(-(t_2-t_1)+\frac{2}{\mu_2}(e^{\mu_2(t_2-t_1)}-1)-\frac{1}{2\mu_2}(e^{2\mu_2(t_2-t_1)}-1)\Big) \nonumber \\
&-\frac{\sigma_2^2}{2\mu_2^2}\Big(-t_1+\frac{2}{\mu_2}\big(e^{\mu_2(t_2-t_1)}+\frac{\mu_2\epsilon_{x_2}}{\kappa_{x_2}}(1-e^{-\kappa_{x_2}(t_2-t_1)})\big)(e^{\mu_2t_1}-1) \nonumber \\ 
&-\frac{1}{2\mu_2}\big(e^{\mu_2(t_2-t_1)}+\frac{\mu_2\epsilon_{x_2}}{\kappa_{x_2}}(1-e^{-\kappa_{x_2}(t_2-t_1)})\big)^2(e^{2\mu_2t_1}-1)\Big) \nonumber \\
&-\frac{1}{\mu_2}\big(\big(e^{\mu_2(t_2-t_1)}+\frac{\mu_2\epsilon_{x_2}}{\kappa_{x_2}}(1-e^{-\kappa_{x_2}(t_2-t_1)})\big)e^{\mu_2t_1} -1\big)\lambda_{x_2}(0)\Big).\label{rho(t1,t2)}
\end{align}
 It remains to derive $\rho(t_1,t_2)$ for $t_1 > t_2$. This can be deduced from the previous formula as demonstrated below. Following \cite{jevtic2017joint}, we introduce the following notation, for $(t_1,t_2)\in [0,T^*]^2$, 
 \[
 \rho(t_1,t_2\ |\ \phi^1,\phi^2):=\rho(t_1,t_2).
 \]
 Thus, applying Theorem \ref{joint density}, for $t_1>t_2$, we have
\begin{align*}
 \rho(t_1,t_2 \ | \ \phi^1,\phi^2)=&{E}_P\Big[ \lambda_{x_2}({t_{2}}) e^{{-\int_{0}^{t_{2}}{\lambda_{x_1}(u) + \lambda_{x_2}(u)}}\mathrm{d}u} {E}_P\Big[ \tilde{\lambda}_{x_1}({t_{1}}) e^{-\int_{t_{2}}^{t_{1}}\tilde{\lambda}_{x_1}(u)\mathrm{d}u} \ | \ \mathcal{K}_{T^*} \Big]\Big],\\
 \rho(t_2,t_1 \ | \ \phi^1,\phi^2)=&{E}_P\Big[ \lambda_{x_1}({t_{2}}) e^{{-\int_{0}^{t_{2}}{\lambda_{x_1}(u) + \lambda_{x_2}(u)}}\mathrm{d}u} {E}_P\Big[ \tilde{\lambda}_{x_2}({t_{1}}) e^{-\int_{t_{2}}^{t_{1}}\tilde{\lambda}_{x_2}(u)\mathrm{d}u} \ | \ \mathcal{K}_{T^*} \Big]\Big],\\
 \rho(t_2,t_1 \ | \ \phi^2,\phi^1)=&{E}_P\Big[ \lambda_{x_2}({t_{2}}) e^{{-\int_{0}^{t_{2}}{\lambda_{x_1}(u) + \lambda_{x_2}(u)}}\mathrm{d}u} {E}_P\Big[ \tilde{\lambda}_{x_1}({t_{1}}) e^{-\int_{t_{2}}^{t_{1}}\tilde{\lambda}_{x_1}(u)\mathrm{d}u} \ | \ \mathcal{K}_{T^*} \Big]\Big].
 \end{align*}
We deduce that  $\rho(t_1,t_2 \ | \ \phi^1,\phi^2)=\rho(t_2,t_1 \ | \ \phi^2,\phi^1)$. Note that $\rho(t_2,t_1 \ | \ \phi^2,\phi^1)$  corresponds to the quantity given in \eqref{rho(t1,t2)}, in which we replace $t_1$ by $t_2$ and $t_2$  by $t_1$, and $\phi^1$  by $\phi^2$ and $\phi^2$ by $\phi^1$. Finally, we obtain that for $t_1>t_2$, $\rho(t_1,t_2)$ is the quantity given in \eqref{rho(t1,t2)} with $t_1$ instead of $t_2$ and $t_2$  instead of $t_1$, and we replace $\phi^1$ by $\phi^2$ and $\phi^2$ by $\phi^1$.
 
The following useful result that will be needed later. For $t\geq 0$,
\begin{eqnarray}
P(\{\tau_{x_1}>t\}\cap\{\tau_{x_2}>t\})= P(\tau_p>t)=\mathbb{E}_P\Big[e^{-\int_0^t\lambda_{x_1}(u)+\lambda_{x_2}(u)\mathrm{d}u}\Big],
\label{Eq0:jointprobability}
\end{eqnarray}
where the last equality follows from \eqref{tau p}. Applying the stochastic Fubini theorem, 
 the integral $\int_0^t\lambda_{x_1}(u)+\lambda_{x_2}(u)\mathrm{d}u$ is a Gaussian random variable with mean
\begin{eqnarray}
m(t)=\frac{\lambda_{x_1}(0)}{\mu_1}(e^{\mu_1t}-1)+\frac{\lambda_{x_2}(0)}{\mu_2}(e^{\mu_2t}-1)\label{m}
\end{eqnarray}
and variance
\begin{align}
\sigma^2(t)=&\Big(\frac{\sigma_1}{\mu_1}\Big)^2\Big(t+\frac{2}{\mu_1}(1-e^{\mu_1t})-\frac{1}{2\mu_1}(1-e^{2\mu_1t})\Big)\notag\\
&+\Big(\frac{\sigma_2}{\mu_2}\Big)^2\Big(t+\frac{2}{\mu_2}(1-e^{\mu_2t})-\frac{1}{2\mu _2}(1-e^{2\mu_2t})\Big). \label{variance}
\end{align} 
Using these fact in \eqref{Eq0:jointprobability} yield to
\begin{equation}
P(\{\tau_{x_1}>t\}\cap\{\tau_{x_2}>t\})=e^{\frac{\sigma^2(t)}{2}-m(t)},\label{Eq1:jointprobability}
\end{equation} 
with $m(t)$ and $\sigma^2(t)$  given in \eqref{m} and \eqref{variance} respectively.

\subsection{Surrender model}
As specified above,  $\tau^s$ denotes the random time at which the policyholder(s) may decide to surrender their policy. Surrender is permitted at discrete time points $t_i$, for $i=1,\dots,K-1$, with the convention that $\{\tau^s = \infty\}$ corresponds to no surrender. We adopt an intensity-based approach to model the surrender behaviour of the policyholder(s). 

In line with common market and academic practice (see, for example, \cite{Kolkiewicz}, \cite{Le Courtois}, \cite{Escobar}, and also \cite{VA1} and references therein for more details), we define the surrender intensity using two components: a deterministic baseline represented by a constant $C>0$, and a stochastic process $D(t)$. The baseline $C$ captures surrender behaviors driven by non-economic factors and personal circumstances, while the process $D(t)$ accounts for additional fluctuations to the baseline due to changes in market conditions.

Consistent with the explanation provided above, we construct the process $D(t)$ based on the spread between the return guaranteed to the policyholder(s) at maturity, assuming no early surrender, and the surrender benefit plus any return offered by reinvesting in the fixed income market.
Therefore, let $Y_t= \log S_t$, where $S_t$ is defined in \eqref{Stock price}. In addition, define $p(t):=- \log \widetilde P(t)$, where $\widetilde P(t)$ is the increasing deterministic function given in \eqref{formula SB}. Then
\begin{equation}
D(t)= Y_{t} -p(t) + \int_{t}^Tf(t,s)ds - \delta T, \quad 0 \leq  t \leq T. \label{Dt}
\end{equation}
Consequently, the surrender intensity is defined as 
\begin{equation}
\lambda^s_t:= \beta |D(t_{i})| + C, \quad t_{i} \le t < t_{i+1}. \label{lambda}
\end{equation}
We note from its expression that $\lambda^s$ is piecewise constant on the interval $[t_{i},t_{i+1})$ for $i=1, \dots, K-1$, and  $\lambda^s_t=0$ for $t\in [0, t_1) \cup [t_K, T]$. 
The constant $\beta \in [0,1]$ captures the impact of the financial market situation on the policyholder(s) decision to terminate the contract. While this approach is conceptually similar to that considered in \cite{VA1}, equation \eqref{lambda} employs the modulus function rather than the parabolic function used in \cite{VA1}. 

Let $\mathcal{G}$ be the $\sigma$-algebra defined as follows
\begin{equation}
\mathcal{G}:= \sigma(\tau^s) \vee \mathcal{F}. \label{G}
\end{equation}

 \begin{assumption}\label{assumption G}
Under the real-world probability $P$, 
the couple $(\tau_{x_1},\tau_{x_2})$, introduced in the previous section,  is independent of $\mathcal{G}$, meaning that mortality  is independent of the financial market and the surrender time.
\end{assumption}

The probability of no surrender is evaluated under the probability $Q\odot P$, introduced in Proposition \ref{QP}, and is given by 
\begin{equation} \label{def:taus}
Q\odot P(\tau^s \geq t_i |\mathcal{F}_{t_i}^{L^1,L^2}) = \exp\Big( -\int_{0}^{t_{i}} \lambda^s_u \mathrm{d}u \Big) ,
\end{equation} 
for all $1 \le  i \le K-1$, and 
\begin{eqnarray}
Q\odot P(\tau^s \geq t |\mathcal{F}_{t}^{L^1,L^2})&=& Q\odot P(\tau^s = \infty |\mathcal{F}_{t}^{L^1,L^2})\notag\\ &=& \exp\Big( -\int_{0}^{t_K} \lambda^s_u \mathrm{d} u \Big),\label{def2:taus}
\end{eqnarray} 
for $t_{K-1} < t$.
Note that the last integral equals $\exp(-\int_{0}^{T} \lambda^s_u du)$ because  $\lambda^s_t=0$ for $t\in  [t_K, T]$.

\section{Pricing the insurance policy}\label{secpricing}
The price  ${\rm P}$ of the insurance policy at time $t=0$ is equal to the sum of the prices of its benefits, that is
  \begin{equation}\label{price contract}
  {\rm P}= {\rm P}^{\rm GMAB} + {\rm P}^{\rm SB} +  {\rm P}^{\rm DB},
  \end{equation}
  where, by using the $Q\odot P$ probability measure 
  \begin{eqnarray*}
  {\rm P}^{\rm GMAB} &=& \mathbb{E}_{Q\odot P}\Big[e^{-\int_0^T r(u)\mathrm{d}u} \,  {\rm GMAB}(T)\Big],\\
  {\rm P}^{\rm SB} &=&\sum_{i=1}^{K-1}  \mathbb{E}_{Q\odot P}\Big[   e^{-\int_0^{t_{i}} r(u)\mathrm{d}u} {\rm SB}(t_{i})\Big],\\
  {\rm P}^{\rm DB} &=& \sum_{i=1}^N  \mathbb{E}_{Q\odot P}\Big[e^{-\int_0^{\bar t_i} r(u)\mathrm{d}u} {\rm DB}(\bar t_i)\Big].
  \end{eqnarray*}

  In fact, as will be seen in the sequel, the $Q\odot P$ probability allows to evaluate the mortality risk under the real-world probability measure $P$, while providing standard risk-neutral valuation under $Q$ for other risks and payoffs dependent on the financial market.

    In the following subsections, we provide explicit formulae for these quantities, as shown in Theorems \ref{Guaranteed benefit}, \ref{surrender}, and \ref{PDB} below. The proofs of these theorems are presented in Section \ref{secproofs}.

\subsection{Price of the guaranteed minimum accumulation benefit}
In this section, we derive the price of the guaranteed minimum accumulation benefit (GMAB). We know from \eqref{formula GMAB} that the payoff of the GMAB is described as:
\begin{equation*}
\text{GMAB}(T)=\max(IS_T,Ie^{\delta T})\mathds{1}_{\{\tau^s>T\}}\mathds{1}_{\{\tau_{x_1}>T\}\cup\{\tau_{x_2}>T\}}.
\end{equation*}
Then, its price is given by
\begin{align}
\text{P}^{\text{GMAB}}=& \mathbb{E}_{Q\odot P}\Big[e^{-\int_0^Tr(u)\mathrm{d}u}\text{GMAB}(T)\Big] \notag\\
=& \mathbb{E}_{Q\odot P}\Big[\mathbb{E}_{Q\odot P}\Big[e^{-\int_0^T r(u)\mathrm{d}u}\text{GMAB}(T)|\mathcal{F}_T^{L^1, L^2}\Big]\Big]\notag\\
=&\mathbb{E}_{Q\odot P}\Big[e^{-\int_0^T r(u)\mathrm{d}u}\max(IS_T,Ie^{\delta T})\mathbb{E}_{Q\odot P}\Big[\mathds{1}_{\{\tau^s>T\}}\mathds{1}_{\{\tau_{x_1}>T\}\cup\{\tau_{x_2}>T\}}|\mathcal{F}_T^{L^1, L^2}\Big]\Big], \label{GMAB new1}
\end{align}
where the last equality follows since $e^{-\int_0^T r(u)du}\max(IS_T,Ie^{\delta T})$ is $\mathcal{F}_T^{L^1, L^2}$-measurable. Note that $\mathcal{F}_T^{L^1, L^2} \subset \mathcal{F} \subset \mathcal{G}$, where $\mathcal{G}$ is defined in \eqref{G}. Therefore, using \eqref{QP2}, \ref{assumption G}, and \eqref{def2:taus}, we have
\begin{eqnarray}
&& \mathbb{E}_{Q\odot P}\Big[\mathds{1}_{\{\tau^s>T\}}\mathds{1}_{\{\tau_{x_1}>T\}\cup\{\tau_{x_2}>T\}}|\mathcal{F}_T^{L^1, L^2}\Big]\notag\\
&=&\mathbb{E}_{Q\odot P}\Big[ \mathbb{E}_{Q\odot P}\Big[\mathds{1}_{\{\tau^s>T\}}\mathds{1}_{\{\tau_{x_1}>T\}\cup\{\tau_{x_2}>T\}}|\mathcal{G}\Big]|\mathcal{F}_T^{L^1, L^2}\Big]\notag\\
&=& \mathbb{E}_{Q\odot P}\Big[\mathds{1}_{\{\tau^s>T\}} \mathbb{E}_{Q\odot P}\Big[\mathds{1}_{\{\tau_{x_1}>T\}\cup\{\tau_{x_2}>T\}}|\mathcal{G}\Big]|\mathcal{F}_T^{L^1, L^2}\Big]\notag\\
&=& \mathbb{E}_{Q\odot P}\Big[\mathds{1}_{\{\tau^s>T\}} E_{ P}\Big[\mathds{1}_{\{\tau_{x_1}>T\}\cup\{\tau_{x_2}>T\}}|\mathcal{G}\Big]|\mathcal{F}_T^{L^1, L^2}\Big]\notag\\
&=& \mathbb{E}_{Q\odot P}\Big[\mathds{1}_{\{\tau^s>T\}} P(\{\tau_{x_1}>T\}\cup\{\tau_{x_2}>T\})|\mathcal{F}_T^{L^1, L^2}\Big]\notag\\
&=& P(\{\tau_{x_1}>T\}\cup\{\tau_{x_2}>T\})\mathbb{E}_{Q\odot P}\Big[\mathds{1}_{\{\tau^s>T\}} |\mathcal{F}_T^{L^1, L^2}\Big]\notag\\
&=&P(\{\tau_{x_1}>T\}\cup\{\tau_{x_2}>T\})e^{-\int_0^{t_K}\lambda^s(u)\mathrm{d}u}. \label{GMAB new2}
\end{eqnarray}
 Substituting \eqref{GMAB new2} into \eqref{GMAB new1}, and using the facts that $e^{-\int_0^T r(u)du}\max(IS_T,Ie^{\delta T}) e^{-\int_0^{t_K}\lambda^s(u)\mathrm{d}u}$ is $\mathcal{F}$-measurable and $Q\odot P$ coincides with $Q$ on $\mathcal{F}$ (see Proposition \ref{QP}), we obtain 
\begin{align*}
\text{P}^{\text{GMAB}}=&P(\{\tau_{x_1}>T\}\cup\{\tau_{x_2}>T\})\mathbb{E}_{Q\odot P}\Big[e^{-\int_0^T r(u)\mathrm{d}u}\max(IS_T,Ie^{\delta T}) e^{-\int_0^{t_K}\lambda^s(u)\mathrm{d}u}\Big]\\
=& P(\{\tau_{x_1}>T\}\cup\{\tau_{x_2}>T\})\mathbb{E}_{Q}\Big[e^{-\int_0^T r(u)\mathrm{d}u}\max(IS_T,Ie^{\delta T}) e^{-\int_0^{t_K}\lambda^s(u)\mathrm{d}u}\Big].
\end{align*}
Let
\begin{equation}
\frac{\mathrm{d}Q^T}{\mathrm{d}Q}=\frac{1}{B(0,T)B(T)},\label{density forward measure}
\end{equation}
such that $Q^T$ is the $T$-forward measure. Then,
\begin{equation*}
\text{P}^{\text{GMAB}}=P(\{\tau_{x_1}>T\}\cup\{\tau_{x_2}>T\})B(0,T)\mathbb{E}_{Q^T}\Big[\max(IS_T,Ie^{\delta T})e^{-\int_0^{t_K}\lambda^s(u)\mathrm{d}u}\Big].
\end{equation*}
Recall that
\[
\max(IS_T,Ie^{\delta T})= Ie^{\delta T}\Big( 1+ \big(\frac{IS_T}{Ie^{\delta T}} -1\big)^+\Big).
\]
From this, we deduce that
\begin{equation*}
\text{P}^{\text{GMAB}}=P(\{\tau_{x_1}>T\}\cup\{\tau_{x_2}>T\})B(0,T)Ie^{\delta T}(A_1+A_2),
\end{equation*}
where
\begin{align*}
A_1:=\mathbb{E}_{Q^T}\Big[e^{-\int_0^{t_{K}}\lambda^s(u)\mathrm{d}u}\Big]\text{ and }
A_2:=\mathbb{E}_{Q^T}\Big[e^{-\int_0^{t_{K}}\lambda^s(u)\mathrm{d}u}\left(\frac{S_T}{e^{\delta T}}-1\right)^+\Big].
\end{align*}
 Observe that 
\begin{equation*}
P_T:=P(\{\tau_{x_1}>T\}\cup\{\tau_{x_2}>T\})=P(\tau_{x_1}>T)+P(\tau_{x_2}>T)-P(\{\tau_{x_1}>T\}\cap\{\tau_{x_2}>T\}).
\end{equation*}
Then applying Theorem \ref{joint density}, we have
\begin{eqnarray*}
P(\tau_{x_1}>T)=\int_T^{T^*}\int_0^{T^*}\rho(t_{1},t_{2})\mathrm{d}t_2\mathrm{d}t_1 \text{ and }
P(\tau_{x_2}>T)=\int_T^{T^*}\int_0^{T^*} \rho(t_{1},t_{2})\mathrm{d}t_1\mathrm{d}t_2.
\end{eqnarray*}
 Using the above and \eqref{Eq1:jointprobability}, we get
 \begin{align}
 P_T:=&P(\{\tau_{x_1}>T\}\cup\{\tau_{x_2}>T\})\notag \\
 =& \int_T^{T^*}\int_0^{T^*}\rho(t_{1},t_{2})\mathrm{d}t_2\mathrm{d}t_1  + \int_T^{T^*}\int_0^{T^*} \rho(t_{1},t_{2})\mathrm{d}t_1\mathrm{d}t_2 - e^{\frac{\sigma^2(T)}{2}-m(T)}, \label{QT}
 \end{align}
 where $m(T)$ and $\sigma^2(T)$ are given in \eqref{m} and \eqref{variance}.
 
We introduce the following notation which will be used in the sequel. Let
 \begin{align} \label{def:wk}
 	w_l & := \int_0^{t_{l}}A(s,T)\mathrm{d}s +\int_0^T f(0,s)\mathrm{d}s  - \delta T  -\omega(t_{l}) -p(t_{l})
 	\intertext{ for $l=1,\dots,K-1$ and }
 	w_K & := \int_0^{T}A(s,T)\mathrm{d}s +\int_0^Tf(0,s)\mathrm{d}s  - \delta T  -\omega(T). \label{def:wk2}
 \end{align}
  Define also $\Delta t_l:=t_l -t_{l-1}$, and $R:=(0,\ldots,0,r)\in \mathbb{R}^K$, with $1<r<2$, and for all $0 \le s \le T$,  and $u \in \R^{K-1}$, $v \in \C^K$ we define the following
  \begin{eqnarray}
  D(u,T) &:=& \exp\Big(\i\sum_{l=1}^{K-1} u_{l}w_l \Big), \notag\\
  \tilde D(v,T) &:=& D(v_1,\dots,v_{K-1},T) \exp\Big(\i v_{K}w_K\Big), \notag\\
  E(s,u,T)&:=& \Sigma_1(s,T)- \i \Sigma_1(s,T) \sum_{l=1}^{K-1}   u_{l} \mathds{1}_{\{0\leq s \leq t_{l}\}}\label{E},\\
  \tilde E(s,v,T)&:=& E(s,v_1,\dots,v_{K-1},T) - \i  \Sigma_1(s,T)v_K,\notag\\
   F(s,u,T)&:=& -\Sigma_2(s,T)+\i (\sigma_2(s)+\Sigma_2(s,T))\sum_{l=1}^{K-1}  u_{l} \mathds{1}_{\{0\leq s \leq t_{l}\}}, \notag\\
   \tilde F(s,v,T) &:=& F(s,v_1,\dots,v_{K-1},T) + \i (\sigma_2(s)+\Sigma_2(s,T)) v_K , \notag
   \end{eqnarray}
   \begin{eqnarray*}
   M(u,T) &:=& D(u,T)\; e^{ \int_0^T\theta^1_s(E(s,u,T))\mathrm{d}s + \int_0^T\theta^2_s(F(s,u,T))\mathrm{d}s }
    \; \prod_{l=2}^K \frac{2\beta \Delta  t_{l}}{u_{l-1}^2 + (\beta \Delta  t_{l})^2} , \notag \\
   N(v,T) &:=& \tilde D(v-\i R,T) \; e^{ \int_0^T\theta^1_s(\tilde E(s,v-\i R,T))\mathrm{d}s + \int_0^T\theta^2_s(\tilde{F}(s,v-\i R,T))\mathrm{d}s } \notag \\
   &&  \times \frac{1}{(\i v_{K} +r-1)(\i v_{K}+r)} 
\prod_{l=2}^K \frac{2\beta \Delta  t_{l}}{v_{l-1}^2 + (\beta \Delta  t_{l})^2} , \text{ for } v\in\mathbb{R}^K. \notag
  \end{eqnarray*}  
 
 The price of the GMAB is given in the following.
 
 \begin{theorem} \label{Guaranteed benefit}
  The price ${\rm P}^{\rm GMAB}$ is given by
  \begin{align*}
  	{\rm P}^{\rm GMAB} = P_TB(0,T)Ie^{\delta T}(A_1+A_2)
  \end{align*}
  with $P_T$ as in \eqref{QT} and
  \begin{eqnarray*}
     A_1 &=&\frac{e^{-C(t_K -t_1)}}{(2\pi)^{K-1 }}e^{-\int_0^T A(s,T)\mathrm{d}s} 
     \int_{\mathbb{R}^{K-1}}M(u,T)\mathrm{d}u, \\
     A_2 &=& \frac{e^{-C(t_K -t_1)}}{(2\pi)^K}e^{-\int_0^T A(s,T)\mathrm{d}s} 
     \int_{\mathbb{R}^{K}} N(u,T) \mathrm{d}u,
   \end{eqnarray*}
   with $A(s,T)$ as in \eqref{dc}.
  \end{theorem}
 \begin{proof}
 	See Section \ref{secproofs}.
 	\end{proof}

\subsection{Price of the surrender benefit}
In this we derive the price of the surrender benefit. We start by recalling that ${\bf t}:=(t_0,t_1, \ldots, t_K)^\top$, with $0=t_0<t_1< \ldots <t_K< T$. Premature surrender is possible at any time point $t_i\in {\bf t}$ with $i = 1, \ldots,K-1.$ From \eqref{formula SB}, we know that the payoff of the surrender benefit at time $t_i\in {\bf t}$ is given by
\begin{equation*}
\text{SB}(t_i)=IS_{t_i}\widetilde P(t_i)\mathds{1}_{\{\tau^s=t_i\}} \mathds{1}_{\{\tau_{x_1}>t_i\}\cup\{\tau_{x_2}>t_i\}}.
\end{equation*}
Therefore, the  price of the surrender benefit is given as follows
\begin{equation*}
\text{P}^{\text{SB}}=\sum_{i=1}^{K-1}\mathbb{E}_{Q\odot P}\Big[e^{-\int_0^{t_i}r(u)\mathrm{d}u}\text{SB}(t_i)\Big].
\end{equation*}
Note that by similar computations to those performed in \eqref{GMAB new1} and \eqref{GMAB new2}, we derive a similar expression for the surrender benefit as follows
\begin{equation}
\mathbb{E}_{Q\odot P}\Big[e^{-\int_0^{t_i}r(u)\mathrm{d}u}\text{SB}(t_i)\Big]= \mathbb{E}_{Q\odot P}\Big[e^{-\int_0^{t_i}r(u)du}IS_{t_i}\widetilde P(t_i)\mathbb{E}_{Q\odot P}\Big[\mathds{1}_{\{\tau^s=t_i\}}|\mathcal{F}_{t_i}^{L^1,L^2}\Big]\Big]P_{t_i} \label{surrender new1}
\end{equation}
where 
\[
P_{t_i}:=P\big(\{\tau_{x_1}>t_i\}\cup\{\tau_{x_2}>t_i\} \big) = P\big(\tau_{x_1}>t_i\big) + P\big(\tau_{x_2}>t_i\big) - P\big(\{\tau_{x_1}>t_i\}\cap\{\tau_{x_2}>t_i\}\big).
\]
Using once more Theorem \ref{joint density} and \eqref{Eq1:jointprobability}, we have
 \begin{equation}
 P_{t_i}= \int_{t_i}^{T^*}\int_0^{T^*}\rho(t_{1},t_{2})\mathrm{d}t_2\mathrm{d}t_1  + \int_{t_i}^{T^*}\int_0^{T^*} \rho(t_{1},t_{2})\mathrm{d}t_1\mathrm{d}t_2 - e^{\frac{\sigma^2(t_i)}{2}-m(t_i)}, \label{Qti}
 \end{equation}
Let $B \in \mathcal{F}_{t_i}^{L^1,L^2}\subset \mathcal{F}_{t_{i+1}}^{L^1,L^2}$ and use \eqref{def:taus} to get,  
\begin{eqnarray*}
\mathbb{E}_{Q\odot P}\left[\mathds{1}_{\{\tau^s=t_i\}}\mathds{1}_{B} \right]&=& \mathbb{E}_{Q\odot P}\Big[ \big(\mathds{1}_{\{t_{i} \le \tau^s\}} -  \mathds{1}_{ \{t_{i+1}\leq \tau^s\}}\big)\mathds{1}_{B} \Big]\\
&=&\mathbb{E}_{Q\odot P}\Big[ \big( e^{-\int_0^{t_{i}} \lambda^s(u) \mathrm{d}u } -  e^{-\int_0^{t_{i+1}} \lambda^s(u) \mathrm{d}u }\big)\mathds{1}_{B} \Big].
\end{eqnarray*}
Since this hold any arbitrary $B \in \mathcal{F}_{t_i}^{L^1,L^2}$, we deduce that 
$$
\mathbb{E}_{Q\odot P}\Big[\mathds{1}_{\{\tau^s=t_i\}}|\mathcal{F}_{t_i}^{L^1,L^2}\Big]=e^{-\int_0^{t_{i}} \lambda^s(u) \mathrm{d}u } -  e^{-\int_0^{t_{i+1}} \lambda^s(u) \mathrm{d}u }.
$$
 Substituting this into \eqref{surrender new1}, we obtain
 \begin{align}
\mathbb{E}_{Q\odot P}\Big[e^{-\int_0^{t_i}r(u)\mathrm{d}u}\text{SB}(t_i)\Big]=& \mathbb{E}_{Q\odot P}\Big[e^{-\int_0^{t_i}r(u)\mathrm{d}u}IS_{t_i}\widetilde P(t_i)\Big(e^{-\int_0^{t_{i}} \lambda^s(u) du } -  e^{-\int_0^{t_{i+1}} \lambda^s(u) \mathrm{d}u }\Big)\Big]P_{t_i}\notag\\
=& \mathbb{E}_{Q}\Big[e^{-\int_0^{t_i}r(u)\mathrm{d}u}IS_{t_i}\widetilde P(t_i)\Big(e^{-\int_0^{t_{i}} \lambda^s(u) \mathrm{d}u } -  e^{-\int_0^{t_{i+1}} \lambda^s(u) \mathrm{d}u }\Big)\Big]P_{t_i}\label{eqsb2},
\end{align}
where the last equality follows from the facts that $Q\odot P$ and $Q$ coincide on $\mathcal{F}$ and the random variable inside the expectations is $\mathcal{F}$-measurable.

Consider $w_l$ as given in \eqref{def:wk} and \eqref{def:wk2}. For $0 \le s \le T$, $i\in \{2,\ldots, K-1\}$, $u \in \R^{i-1}$, and $v \in \R^{i}$, define
 \begin{align}
    D^{i}(u,T) :=& \exp\Big( \i\sum_{l=1}^{i-1} u_{l}w_l - \omega(t_{i}) \Big), \notag\\
    \tilde D^{i}(v,T) :=& D^{i}(v_1,\ldots,v_{i-1},T)  \exp\Big( \i\  v_{i}w_i  \Big),\notag\\
  	E^{i}(s,u,T):=&  - \i\Sigma_1(s,T)\sum_{l=1}^{i-1} u_{l} \mathds{1}_{\{0\leq s \leq t_{l}\}} 
	  ,\notag \\
	\tilde E^{i}(s,v,T) :=& E^{i}(s,v_1,\ldots,v_{i-1},T) - \i   \Sigma_1(s,T)v_{i}, \notag\\\label{Ei}
  	F^{i}(s,u,T):=& \i(\sigma_2(s)+\Sigma_2(s,T))\sum_{l=1}^{i-1} u_{l} \mathds{1}_{\{0\leq s \leq t_{l}\}}  + \sigma_2(s), \\
  	\tilde F^{i}(s,v,T) :=& F^{i}(s,v_1,\ldots,v_{i-1},T) + \i (\sigma_2(s)+\Sigma_2(s,T))v_{i},  \notag \\
  	M^{i}(u,T) :=& D^{i}(u,T) \;	\exp\Big(  \int_0^{t_{i}} \Big( \theta^1_s( E^{i}(s,u,T) )+ \theta^2_s ( F^{i}(s,u,T)) \Big)  \mathrm{d}s  \Big) \notag \\
  	&  \times \prod_{l=2}^{i} \frac{2\beta \Delta  t_{l}}{u_{l-1}^2 + (\beta \Delta  t_{l})^2}, \notag \\
    N^{i}(v,T) :=& \tilde D^{i}(v,T) \;	\exp\Big(  \int_0^{t_{i}} \Big( \theta^1_s( \tilde E^{i}(s,v,T) )+ \theta^2_s ( \tilde F^{i}(s,v,T)) \Big) \mathrm{d}s  \Big) \notag \\
  	&  \times  \prod_{l=2}^{i+1} \frac{2\beta \Delta  t_{l}}{v_{l-1}^2 + (\beta \Delta  t_{l})^2},\notag
  	\end{align}
and for $v \in \R$, define
\begin{eqnarray*}
	\tilde E^{1}(s,v,T) &:=&  - \i   \Sigma_1(s,T)v, \\
	\tilde F^{1}(s,v,T) &:=& \sigma_2(s) + \i (\sigma_2(s)+\Sigma_2(s,T))v,\\
	 N^{1}(v,T) &:=&  \;\exp\big( \i v w_1 - \omega(t_{1})\big)	\exp\Big(  \int_0^{t_{1}} \Big( \theta^1_s( \tilde E^{1}(s,v,T) )+ \theta^2_s ( \tilde F^{1}(s,v,T)) \Big)  \mathrm{d}s  \Big)  \\
  	&&  \times   \frac{2\beta \Delta  t_{2}}{v^2 + (\beta \Delta  t_{2})^2}.  
	\end{eqnarray*}
The price of the SB is given in the following.
\begin{theorem}\label{surrender} 
   The price ${\rm P}^{\rm SB}$ of the surrender benefit is given by  
   \begin{eqnarray*} 
   {\rm P}^{\rm SB} &=& I \sum_{i=1}^{K-1} \widetilde P(t_{i}) (B_{i}^{1} - B_{i}^{2})P_{t_i}, 
   \end{eqnarray*}
   where $P_{t_i}$ is as in \eqref{Qti}, $B_1^{1}=1$ and for $i\in \{2, \ldots, K-1\}$,
   \begin{eqnarray*}
      B_{i}^{1}&=& \frac{ e^{-C(t_{i} -t_1)} }{ (2\pi)^{i-1}} 
      \int_{\mathbb{R}^{i-1}} M^{i}(u,T)  \mathrm{d}u, 
 \end{eqnarray*}
  and, for $i\in\{1, \ldots, K-1\},$ 
 \begin{eqnarray*}
      B_{i}^{2}&=&  \frac{ e^{-C(t_{i+1} -t_1)} }{ (2\pi)^{i}} 
      \int_{\mathbb{R}^{i}} N^{i}(u,T)  \mathrm{d}u.
 \end{eqnarray*}
\end{theorem}

\subsection{Price of the death benefit}

     As for the death benefit, we still have that $\bar{\bf t}:=(\bar t_0,\bar{t}_1, \ldots, \bar{t}_N)^\top$ with $0=\bar t_0 < \bar{t}_1< \ldots < \bar{t}_N=T$. The mortality of the two members of the couple is monitored by the insurer at the time points ${\bar t_i}$. 
      We know from \eqref{formula DB} that the death benefit at time $\bar t_i \in \bar{\bf t}$ is given by   
      \begin{align}
\text{DB}(\bar{t}_i)=&\max(IS_{\bar{t}_i},Ie^{\delta \bar{t}_i})\mathds{1}_{\{\tau^s\geq \bar{t}_i\}}[\mathds{1}_{\{\bar{t}_{i-1}\leq\tau_{x_1}<\bar{t}_i\}}+\mathds{1}_{\{\bar{t}_{i-1}\leq\tau_{x_2}<\bar{t}_i\}}\notag\\ &+(\alpha-2)\mathds{1}_{\{\bar{t}_{i-1}\leq\tau_{x_1}<\bar{t}_i\}}\mathds{1}_{\{\bar{t}_{i-1}\leq\tau_{x_2}<\bar{t}_i\}}]\notag\\
=& \max(IS_{\bar{t}_i},Ie^{\delta \bar{t}_i})\mathds{1}_{\{\tau^s\geq \bar{t}_i\}} X_{\alpha, i}\label{definition death benefit}
\end{align}
where $1<\alpha<2$, and we used obvious notation in the last line. Therefore, the price of the death benefit is given as follows:
\begin{align}\label{eqdb1}
\text{P}^{\text{DB}}=&\sum_{i=1}^N \mathbb{E}_{Q\odot P}\Big[e^{-\int_0^{\bar{t}_i}r(u) \mathrm{d}u}\text{DB}(\bar{t}_i)\Big]\notag\\
=& \sum_{i=1}^N \mathbb{E}_{Q\odot P}\Big[e^{-\int_0^{\bar{t}_i}r(u) \mathrm{d}u}\max(IS_{\bar{t}_i},Ie^{\delta \bar{t}_i})\mathds{1}_{\{\tau^s\geq \bar{t}_i\}} X_{\alpha, i}\Big].
\end{align}
For $i\in\{1, \ldots N\}$, we define $P(i)$ by
\begin{align*}
P(i):=&\mathbb{E}_P(X_{\alpha,i})\\ 
=&P\big(\bar{t}_{i-1}\leq\tau_{x_1}<\bar{t}_i \big) + P\big(\bar{t}_{i-1}\leq\tau_{x_2}<\bar{t}_i \big)+(\alpha -2) P\big(\{\bar{t}_{i-1}\leq\tau_{x_1}<\bar{t}_i\}\cap \{\bar{t}_{i-1}\leq\tau_{x_2}<\bar{t}_i\}\big).
\end{align*}
Using Theorem \ref{joint density}, we have
\begin{align*}
 &P\big(\bar{t}_{i-1}\leq\tau_{x_1}<\bar{t}_i \big)=\int_{\bar{t}_{i-1}}^{\bar{t}_i}\int_0^{T^*}\rho(t_{1},t_{2}) \mathrm{d}t_2 \mathrm{d}t_1, \,\,\,\,\,\,\, P\big(\bar{t}_{i-1}\leq\tau_{x_2}<\bar{t}_i \big)=\int_{\bar{t}_{i-1}}^{\bar{t}_{i}}\int_0^{T^*} \rho(t_{1},t_{2}) \mathrm{d}t_1 \mathrm{d}t_2,\\
& P\big(\{\bar{t}_{i-1}\leq\tau_{x_1}<\bar{t}_i\}\cap \{\bar{t}_{i-1}\leq\tau_{x_2}<\bar{t}_i\}\big)=\int_{\bar{t}_{i-1}}^{\bar{t}_i}\int_{\bar{t}_{i-1}}^{\bar{t}_i}\rho(t_{1},t_{2}) \mathrm{d}t_1 \mathrm{d}t_2.
\end{align*}
 Substituting this in the expression of $P(i)$ yields
 \begin{align}
 P(i)=& \int_{\bar{t}_{i-1}}^{\bar{t}_i}\int_0^{T^*}\rho(t_{1},t_{2}) \mathrm{d}t_2 \mathrm{d}t_1 + \int_{\bar{t}_{i-1}}^{\bar{t}_{i}}\int_0^{T^*} \rho(t_{1},t_{2}) \mathrm{d}t_1 \mathrm{d}t_2\notag\\
 & +(\alpha -2)\int_{\bar{t}_{i-1}}^{\bar{t}_i}\int_{\bar{t}_{i-1}}^{\bar{t}_i}\rho(t_{1},t_{2}) \mathrm{d}t_1 \mathrm{d}t_2. \label{Qi}
 \end{align}

Consider $w_l$ given by \eqref{def:wk} and \eqref{def:wk2}, and in addition, for $i\in \{1, \ldots, N\}$, define
\begin{equation}
	w_{\bar t_i}:=   \int_0^{\bar t_i}A(s,\bar t_i) \mathrm{d}s +\int_0^{\bar t_i}f(0,s) \mathrm{d}s-\omega(\bar t_i) . \label{def:wk3}
 \end{equation}
  Recall that $\Delta t_l:=t_l -t_{l-1}$. For all $j\in \{1,\ldots, K-1\}$, let $R:=(0,\ldots,0,r)\in \mathbb{R}^{j+1}$, with $1<r<2$, $u \in \R^{j}$, and $v \in \C^{j+1}$.
For all $0 \le s \le T$, for all $j\in \{1,\ldots, K-2\}$ and all $i\in \{1, \ldots, N\}$ such that $t_j < \bar t_i \le t_{j+1}$, and for $j=K-1$ and all $i$ such that    $t_{K-1} < \bar t_i \leq T$, we define the following quantities
\begin{align}
D^{j}(u,T) :=& \exp\Big( \i\sum_{l=1}^{j} u_{l}w_l    \Big),\notag\\
\tilde D^{j,i}(v,T) :=& D^{j}(v_1,\dots,v_{j},T) \exp\Big(\i v_{j+1}w_{\bar t_i}\Big), \notag\\
E_{j,i}(s,u,T):=& \Sigma_1(s,\bar t_i) -\i\Sigma_1(s,T)\sum_{l=1}^{j} u_{l} \mathds{1}_{\{0\leq s \leq t_{l}\}} 
	  ,\notag\\
	 \tilde E_{j,i}(s,v,T):=& E_{j,i}(s,v_1,\dots,v_{j},T) - \i \Sigma_1(s,\bar t_i)v_{j+1},\notag\\
F_{j,i}(s,u,T):=& -\Sigma_2(s,\bar t_i)+ \i(\sigma_2(s)+\Sigma_2(s,T))\sum_{l=1}^{j} u_{l} \mathds{1}_{\{0\leq s \leq t_{l}\}} , \label{DB notation}\\
 \tilde F_{j,i}(s,v,T):=& F_{j,i}(s,v_1,\dots,v_{j},T) + \i  (\sigma_2(s)+\Sigma_2(s, \bar t_i))v_{j+1} , \notag
 \end{align}
\begin{align*}
M^{j,i}(u,T):=& D^{j}(u,T) \; \exp\Big(  \int_0^{\bar t_{i}} \Big( \theta^1_s( E_{j,i}(s,u,T) )+ \theta^2_s ( F_{j,i}(s,u,T)) \Big)  \mathrm{d}s  \Big)\notag\\
& \times \prod_{l=2}^{j+1}  \frac{2\beta \Delta  t_{l}}{u_{l-1}^2 + (\beta \Delta  t_{l})^2}, \notag\\
N^{j,i}(v,T):=& \tilde D^{j,i}(v-\i R,T) \; \exp\Big(  \int_0^{\bar t_{i}} \Big( \theta^1_s( \tilde E_{j,i}(s,v- \i R,T) )+ \theta^2_s ( \tilde F_{j,i}(s,v- \i R,T)) \Big)  \mathrm{d}s  \Big)\notag\\
& \times\frac{\exp\big(-\delta \bar t_i (\i v_{j+1} +r)\big)}{(\i v_{j+1} +r-1)(\i v_{j+1}+r)}\prod_{l=2}^{j+1}  \frac{2\beta \Delta  t_{l}}{v_{l-1}^2 + (\beta \Delta  t_{l})^2}, \text{ for } v \in \mathbb{R}^{j+1}. \notag
\end{align*}

Further, we define  for $0\leq s\leq T$, $u\in \mathbb{R}$ and $i\in\{1,\ldots,N\}$ such that $\bar t_i \leq t_1$,
\begin{align*}
E_i(s,u):=& \Sigma_1(s, \bar t_i)- (r+\i u)\Sigma_1(s, \bar t_i),\\
F_i(s,u) :=& -\Sigma_2(s,\bar t_i) +(r+ \i u) (\sigma_2(s)+\Sigma_2(s,\bar t_i)),\\
N^i(u):=& \exp\Big((r +\i u)w_{\bar t_i} \Big)\;\exp\Big( \int_0^{\bar t_i}\theta^1_s(E_i(s,u)) \mathrm{d}s + \int_0^{\bar t_i}\theta^2_s(F_i(s,u)) \mathrm{d}s\Big)\\
& \times \frac{\exp\big(-\delta \bar t_i(\i u +r)\big)}{(\i u +r-1)(\i u+r)}.
\end{align*}
 The price of the DB is given in the following.
\begin{theorem}\label{PDB}
The price ${\rm P}^{\rm DB}$ of the death benefit is given by
\begin{align*} 
{\rm P}^{\rm DB} =& \sum_{i:\: \bar t_i \le t_1}P(i)( Ie^{\delta \bar{t}_i}B(0, \bar t_i) + Ie^{\delta \bar{t}_i}B(0, \bar t_i) A_{0,i})\\
& + \sum_{j=1}^{K-2}\sum_{i:\:\bar t_i\in(t_j,t_{j+1}]}  P(i)Ie^{\delta \bar{t}_i}B(0,\bar t_i)\big(A^1_{j,i} + A^2_{j,i}\big)\\
& + \sum_{i:\: \bar t_i\in (t_{K-1},T]} P(i)Ie^{\delta \bar{t}_i}B(0,\bar t_i)\big(A^1_{K-1,i} + A^2_{K-1,i}\big),
\end{align*}
where $P(i)$ is as in \eqref{Qi}, 
and for $j\in \{1,\ldots, K-1\}$
\begin{align*}
&A_{0,i}=\frac{e^{-\int_0^{\bar t_i} A(s,\bar t_i)\mathrm{d}s}}{2\pi}\int_{\mathbb{R}} N^i(u)\mathrm{d}u,\,\,\,\,\
A_{j,i}^1= \frac{e^{-C(t_{j+1} -t_1)}}{(2\pi)^{j}}e^{-\int_0^{\bar t_i} A(s,\bar t_i)\mathrm{d}s }\int_{\mathbb{R}^{j}}M^{j,i}(u,T) \mathrm{d}u,\\
&A_{j,i}^2= \frac{e^{-C(t_{j+1} -t_1)}}{(2\pi)^{j+1}}e^{-\int_0^{\bar t_i} A(s,\bar t_i)\mathrm{d}s }\int_{\mathbb{R}^{j+1}}N^{j,i}(u,T) \mathrm{d}u.
\end{align*}
\end{theorem}
where $M^{j,i}, N^{j,i}$ and $N^i$ are described in \eqref{DB notation}.

\section{Proofs}\label{secproofs}
This section is devoted to the proofs of the main results.
\subsection{Proof of Theorem \ref{Guaranteed benefit}}
	We have seen above that 
	\begin{equation*}
		\text{P}^{\text{GMAB}}=P(\{\tau_{x_1}>T\}\cup\{\tau_{x_2}>T\})B(0,T)Ie^{\delta T}(A_1+A_2),
	\end{equation*}
	where
	\begin{align*}
		A_1:=\mathbb{E}_{Q^T}\Big[e^{-\int_0^{t_{K}}\lambda^s(u)\mathrm{d}u}\Big] \text{ and }
		A_2:=\mathbb{E}_{Q^T}\Big[e^{-\int_0^{t_{K}}\lambda^s(u)\mathrm{d}u}\Big(\frac{S_T}{e^{\delta T}}-1\Big)^+\Big].
	\end{align*}
	We start by calculating $A_1$ and $A_2$. Using the definition of $\lambda^s$ in \eqref{lambda}, we get
	\begin{eqnarray*}
		A_1 = e^{-C(t_K -t_1)} \mathbb{E}_{Q^T}\Big[\prod_{i=2}^K e^{-\beta \Delta t_i |D(t_{i-1})|}\Big]
		= e^{-C(t_K -t_1)} \mathbb{E}_{Q^T}\Big[f(D(t_1), \ldots, D(t_{K-1}))\Big],
	\end{eqnarray*}
	where 
	$f(x_1, \ldots, x_{K-1})
	:=\prod_{l=2}^{K} \big( e^{-\beta \Delta t_{l} |x_{l-1}|}\big)
	=\prod_{l=2}^{K} f_{l}(x_{l-1}).$ 
	Let $\hat f$ denote  the Fourier transform of a generic function $f$. Then, for any $y\in \mathbb{C}$, 
	\begin{eqnarray}
		\hat{f_l}(y)= \int_{\mathbb{R}}e^{\i yt}e^{-\beta \Delta t_l |t|}\mathrm{d}t 
		= \frac{2\beta \Delta  t_{l}}{y^2 + (\beta \Delta  t_{l})^2} . \label{Fourier 1}
	\end{eqnarray}
	Thus we deduce that
	\begin{equation}\label{eq:fhat}
		\hat{f}(y_1, \ldots, y_{K-1})=  \prod_{l=2}^K \frac{2\beta \Delta  t_{l}}{y_{l-1}^2 + (\beta \Delta  t_{l})^2},
	\end{equation}
	and we can see that $\hat{f}\in L^1(\mathbb{R}^{K-1})$. Thanks to \cite[Theorem 3.2]{Eberlein2010}, it holds
	\begin{equation}
		\mathbb{E}_{Q^T}\Big[f(D(t_1), \ldots, D(t_{K-1}))\Big] = \frac{1}{(2\pi)^{K-1}}\int_{\mathbb{R}^{K-1}}\tilde M(\i u)\hat{f}(-u)\mathrm{d}u, \label{Price-Fourier}
	\end{equation}
	where for any $u=(u_1, \ldots, u_{K-1})$, $\tilde M(\i u)$ is defined as follows
	\begin{equation*}
		\tilde M(\i u):=\mathbb{E}_{Q^T}\Big[e^{\i u_1 D(t_1) + \ldots + \i u_{K-1}D(t_{K-1})}\Big].
	\end{equation*}
	Using the representation \eqref{D:appendix} in Appendix \ref{app:prelim}, we get
	\begin{eqnarray*}
		\tilde M(\i u) &=& \exp\Big[\i\sum_{l=1}^{K-1} u_{l}\Big(-p(t_{l}) - \delta T+ \int_0^T f(0,s)\mathrm{d}s + \int_0^{t_{l}}A(s,T)\mathrm{d}s -\omega(t_{l})\Big)\Big]\\
		&& \quad \times \mathbb{E}_{Q^T}\Big[\exp\Big(\i\sum_{l=1}^{K-1} u_l \Big(\int_0^{t_{l}} (\sigma_2(s)+\Sigma_2(s,T))\mathrm{d}L^2_s - \int_0^{t_{l}}\Sigma_1(s,T)\mathrm{d}L^1_s\Big)\Big)\Big].
	\end{eqnarray*}
	Using the representation \eqref{Bank ac}, the density of $Q^T$ given in \eqref{density forward measure} can be rewritten as follows
	\begin{equation*} 
		\frac{dQ^T}{dQ}= \exp\Big( -\int_0^T A(s,T)\mathrm{d}s + \int_0^T \Sigma_1(s,T)\mathrm{d}L^1_s -\int_0^T \Sigma_2(s,T)\mathrm{d}L^2_s \Big).
	\end{equation*}
	Therefore
	\begin{eqnarray*}
		\lefteqn{ \mathbb{E}_{Q^T}\Big[\exp\Big( \i\sum_{l=1}^{K-1} \int_0^{t_{l}} (u_{l}(\sigma_2(s)+\Sigma_2(s,T)))\mathrm{d}L^2_s - \int_0^{t_{l}}u_{l}\Sigma_1(s,T)\mathrm{d}L^1_s\Big)\Big]} \hspace{1cm}\\
		&=& \mathbb{E}_Q\Big(\exp\Big(- \int_0^TA(s,T)\mathrm{d}s + \int_0^T\Sigma_1(s,T))\mathrm{d}L^1_s -\int_0^T\Sigma_2(s,T)\mathrm{d}L^2_s  \Big) \\
		&& \qquad \quad \times \exp\Big( \i\sum_{l=1}^{K-1} (\int_0^{t_{l}} u_{l}(\sigma_2(s)+\Sigma_2(s,T)))\mathrm{d}L^2_s - \int_0^{t_{l}}u_{l}\Sigma_1(s,T)\mathrm{d}L^1_s)\Big)\Big]\\
		&=& e^{-\int_0^T A(s,T)\mathrm{d}s}\mathbb{E}_Q\Big[\exp\Big( \int_0^T E(s,u,T)\mathrm{d}L^1_s + \int_0^T F(s,u,T)\mathrm{d}L^2_s\Big)\Big],
	\end{eqnarray*}
	where $E$ and $F$ are given in \eqref{E}. In virtue of \eqref{AssM1} and \eqref{AssM22}, the above expectation exists. By the independence of $L^1$ and $L^2$ and using equation \eqref{cumulant}, the last quantity simplifies to the following:
	\[
	\exp\Big( -\int_0^T A(s,T)\mathrm{d}s + \int_0^T\theta^1_s(E(s,u,T))\mathrm{d}s + \int_0^T\theta^2_s(F(s,u,T))\mathrm{d}s \Big).
	\]
	Consequently, using the definition of $D(u,T)$ \eqref{E}, we have
	\begin{eqnarray*}
		\tilde M(\i u)&=& D(u,T)  \exp\Big(-\int_0^T A(s,T)\mathrm{d}s + \int_0^T\theta^1_s(E(s,u,T))\mathrm{d}s + \int_0^T\theta^2_s(F(s,u,T))\mathrm{d}s \Big).
	\end{eqnarray*}
	Therefore, the representation for $A_1$ follows using \eqref{Price-Fourier}.

	It remains to calculate $A_2$. 

It follows from \eqref{Dt} that, $D(T)=Y_T -\delta T$ and since $S_T=e^{Y_T}$, we have  $\frac{S_T}{e^{\delta T}} 
= e^{D(T)}$.

Consequently, by the same arguments used previously in the calculation of $A_1$, we have 
	\begin{eqnarray*}
		A_2 &=& e^{-C(t_K -t_1)} \mathbb{E}_{Q^T}\Big[f(D(t_1), \ldots, D(t_{K-1}))\Big(e^{D(T)}-1\Big)^+\Big]\\
		&=& e^{-C(t_K -t_1)} \mathbb{E}_{Q^T}\Big[F(D(t_1), \ldots, D(t_{K-1}), D(T))\Big],
	\end{eqnarray*}
	where
	\[
	F(x_1, \ldots, x_K):=f(x_1,\dots,x_{K-1})  (e^{x_K }-1)^+ .
	\] 
	To ensure integrability, we define $g(x_1,\ldots, x_K):=F(x_1, \ldots, x_K) e^{-rx_K}$, with $1<r<2$. Let
	\[
	g_K (x_K) :=  (e^{x_K }-1)^+ e^{-r x_K}. 
	\]
	Thus, $g_K \in L^1(\mathbb{R})$, such that $g\in L^1(\mathbb{R}^{K})$.
	Elementary computations shows that
	for all $ y \in \mathbb{R}$
	\[
	\hat{g}_K(y)= \frac{1}{(\i y -r+1)(\i y-r)}. 
	\]
	Note that $|\hat{g}_K(y)|_{\mathbb{C}}= (((1-r)^2+ y^2)(r^2+y^2))^{-1/2}$, and thus, $\hat{g}_K\in L^1(\mathbb{R})$. Consequently, combining this with \eqref{eq:fhat}, we deduce that $\hat{g}\in L^1(\mathbb{R}^{K})$, and 
	\begin{equation}
		\hat{g}(y_1, \ldots, y_K)=  \frac{1}{(\i y_{K} -r+1)(\i y_{K}-r)} \prod_{l=2}^K \frac{2\beta \Delta  t_{l}}{y_{l-1}^2 + (\beta \Delta  t_{l})^2}. \label{Fourier 2}
	\end{equation}
	Since $g, \hat{g}\in L^1(\mathbb{R}^{K})$, applying \cite[Theorem 3.2]{Eberlein2010}, we obtain 
	\begin{equation}
		\mathbb{E}_{Q^T}\Big[F(D(t_1), \ldots, D(t_{K-1}), D(T))\Big] = \frac{1}{(2\pi)^{K}}\int_{\mathbb{R}^{K}}\tilde N(R+\i u)\hat{F}(\i R-u)\mathrm{d}u, \label{Price-Fourier'}
	\end{equation}
	where $R:=(0, \ldots, 0, r)\in\mathbb{R}^K$, $1<r<2$, and $\tilde N(R+\i u)$ is defined as follows 
	\begin{equation*}
		\tilde N(R+\i u):=\mathbb{E}_{Q^T}\Big[e^{\i u_1 D(t_1) + \ldots + \i u_{K-1}D(t_{K-1})+ (\i u_K +r)D(T)}\Big].
	\end{equation*}
	Using once more the notation from \eqref{E}, we get
	\begin{eqnarray*}
		\tilde N(R+\i u)
		&=& \tilde D(u -\i R ,T)
		e^{-\int_0^T A(s,T)ds}  \\
		&& \quad \times \mathbb{E}_Q\Big[\exp\Big( \int_0^{T} \tilde E(s,u-\i R,T) \mathrm{d}L^1_s + \int_0^{T} \tilde F(s,u-\i R,T) \mathrm{d}L^2_s \Big)\Big].
	\end{eqnarray*}
	Since $1<r <2$, and in virtue of \eqref{AssM1}, \eqref{AssM22} and \eqref{AssM2}, we have $|(1-r)\Sigma_1(s,T)| \leq M_1$, and  $|r\sigma_2(s) +(r-1)\Sigma_2(s,T)| \leq 2 \frac{M_2}{3} + (r-1)\frac{M_2}{3} \leq M_2$. Therefore, the above expectation exists.
	By the independence of $L^1$ and $L^2$, and \eqref{cumulant}, we deduce that
	\begin{eqnarray}
		\tilde N(R+\i u)
		&=& \tilde D(u-\i R  ,T)  e^{-\int_0^T A(s,T)ds}\notag \\
		&& \times \exp\Big( \int_0^T \theta^1_s(\tilde{E}(s,u-\i R,T)) \mathrm{d}s + \int_0^T \theta^2_s(\tilde{F}(s,u-\i R,T))    \mathrm{d}s \Big). \label{M'}
	\end{eqnarray}
	Observe that for any $u\in \mathbb{R}^{K}$, we have 
	\[
	\hat{g}(u)=\int_{\mathbb{R}^{K}}e^{\i \langle u, x \rangle } e^{- \langle R , x \rangle} F(x)\mathrm{d}x = \hat{F}(u+\i R). 
	\]
	Consequently, we have that
	\begin{eqnarray}
		\hat{F}(\i R - u) =\hat{g}(-u) 
		&=& \frac{1}{(\i u_{K} +r-1)(\i u_{K}+r)} \prod_{l=2}^K \frac{2\beta \Delta  t_{l}}{u_{l-1}^2 + (\beta \Delta  t_{l})^2}. \label{Fourier 3}
	\end{eqnarray}
	The claim follows by substituting \eqref{M'} and \eqref{Fourier 3} into \eqref{Price-Fourier'}.
	
\subsection{Proof of Theorem \ref{surrender}}

	We know from \eqref{eqsb2} that 
		\begin{eqnarray*}
			\mathbb{E}_{Q\odot P}\Big[e^{-\int_0^{t_i}r(u)\mathrm{d}u}\text{SB}(t_i)\Big]
			&=& I\widetilde P(t_i) \mathbb{E}_{Q}\Big[e^{-\int_0^{t_i}r(u)\mathrm{d}u}S_{t_i}\Big(e^{-\int_0^{t_{i}} \lambda^s(u) \mathrm{d}u } -  e^{-\int_0^{t_{i+1}} \lambda^s(u) \mathrm{d}u }\Big)\Big]P_{t_i}
		\end{eqnarray*}
		But
		\begin{align*}
			&\mathbb{E}_{Q}\Big[e^{-\int_0^{t_i}r(u)du}S_{t_i}\Big(e^{-\int_0^{t_{i}} \lambda^s(u) \mathrm{d}u } -  e^{-\int_0^{t_{i+1}} \lambda^s(u) \mathrm{d}u }\Big)\Big]\\
			=&  \mathbb{E}_{Q}\Big[  e^{-\int_0^{t_{i}}r_u \mathrm{d}u} S_{t_{i}} e^{-\int_0^{t_{i}} \lambda^s(u) \mathrm{d}u }      \Big] 
			-  \mathbb{E}_{Q}\Big[  e^{-\int_0^{t_{i}}r_u \mathrm{d}u} S_{t_{i}} e^{-\int_0^{t_{i+1}} \lambda^s(u) \mathrm{d}u }      \Big]=:B_i^1+B_i^2.
		\end{align*}
		To compute $B^1_{i}$ and $B^2_{i}$, we use the spot probability measure $Q^{S,i},\ i=1,\dots,K-1$ defined by its Radon-Nikodym derivative
		\begin{equation}
			\frac{\mathrm{d}Q^{S,i}}{\mathrm{d}Q} = e^{-\int_0^{t_{i}}r_u \mathrm{d}u} S_{t_{i}}. \label{bis forward measure}
		\end{equation}
		Note that equation \eqref{bis forward measure} defines a density process since the discounted stock price $\big(e^{-\int_0^{t}r_u du}S(t)\big)_t$ is a $Q$-martingale. Using this new measure, we deduce the following  
		\begin{align}
			B^1_i =&\mathbb{E}_{Q}\Big[  e^{-\int_0^{t_{i}}r_u \mathrm{d}u}  S_{t_{i}}  e^{-\int_0^{t_{i}} \lambda^s(u) \mathrm{d}u }    \Big] = \mathbb{E}_{Q^{S,i}}\Big[ e^{-\int_0^{t_{i}} \lambda^s(u) \mathrm{d}u }    \Big], \label{B^1} \\
			B^2_i =&\mathbb{E}_{Q}\Big[  e^{-\int_0^{t_{i}}r_u du}  S_{t_{i}} e^{-\int_0^{t_{i+1}} \lambda^s(u) du }     \Big] =  \mathbb{E}_{Q^{S,i}}\Big[ e^{-\int_0^{t_{i+1}} \lambda^s(u) \mathrm{d}u }    \Big]. \label{B^2}
		\end{align}
		We start by computing $B_i^1$. By its construction,  $\lambda^s(u)=0$ for $u\in[0, t_1)$, so $B_1^{1}=1$.
		For $i\in \{2, \ldots, K-1\}$, the computation of $B^1_i$ is  analogous to the computation of $A_1$ in the proof of Theorem \ref{Guaranteed benefit}. In fact, similar to  \eqref{Price-Fourier}, we have
		
		It holds
		\begin{align}   		
			e^{C (t_{i} - t_1) } \mathbb{E}_{Q^{S,i}} \Big[ e^{-\int_0^{t_{i}} \lambda^s(u) \mathrm{d}u }    \Big] 
			=& \mathbb{E}_{Q^{S,i}}\Big[f(D(t_1), \ldots, D(t_{i-1}))\Big]\notag\\
			=&	\frac{1}{(2\pi)^{i-1}}\int_{\mathbb{R}^{i-1}}\tilde M^{i-1}(\i u)\hat{f}(-u)\mathrm{d}u, \label{B1}
		\end{align}
		where 
		\begin{align*}
		f(x_1, \ldots, x_{i-1}):=&\prod_{l=2}^{i} e^{-\beta \Delta t_{l} |x_{l-1}|}, \text{ with } \Delta t_l= t_l  - t_{l-1}; \,\,\,
			\hat{f}(u_1, \ldots, u_{i-1})=  \prod_{l=2}^{i} \frac{2\beta \Delta  t_{l}}{u_{l-1}^2 + (\beta \Delta  t_{l})^2} ,\\
			\tilde M^{i-1}(\i u) =& \mathbb{E}_{Q^{S,i}}\Big[e^{\i u_1 D(t_1) + \ldots + \i u_{i-1}D(t_{i-1})}\Big].
		\end{align*}

		Using \eqref{D:appendix} we deduce  that
		\begin{eqnarray*}
			\tilde M^{i-1}(\i u) &=&
			\mathbb{E}_{Q}\Big[e^{\i u_1 D(t_1) + \ldots + \i u_{i-1}D(t_{i-1})+ \int_0^{t_{i}} \sigma_2(s)\mathrm{d}L^2_s  -\omega(t_{i})}\Big]\\
			&=& \exp\Big( \i\sum_{l=1}^{i-1} u_{l}(-p(t_{l}) - \delta T+ \int_0^T f(0,s)\mathrm{d}s + \int_0^{t_{l}}A(s,T)\mathrm{d}s -\omega(t_{l})) \Big)\\
			&& \times \mathbb{E}_{Q}\Big[
			\exp\Big( \i\sum_{l=1}^{i-1} \Big(\int_0^{t_{l}} u_{l}(\sigma_2(s)+\Sigma_2(s,T))\mathrm{d}L^2_s - \int_0^{t_{l}}u_{l} \Sigma_1(s,T)\mathrm{d}L^1_s\Big) \\
			&&   \qquad \qquad \qquad +  \int_0^{t_{i}} \sigma_2(s)\mathrm{d}L^2_s  -\omega(t_{i}) \Big)   \Big].
		\end{eqnarray*}
		Consider  $ E^{i}(s,u,T)$ and $F^{i}(s,u,T)$ defined in \eqref{Ei}. Using \eqref{cumulant}, we have   
		\begin{align*}
		& \mathbb{E}_{Q}\Big[\exp\Big(  \int_0^{t_{i}} E^{i}(s,u,T) \mathrm{d}L^1_s + \int_0^{t_{i}} F^{i}(s,u,T) \mathrm{d}L^2_s \Big) \Big] \\
			= &\exp\Big(  \int_0^{t_{i}} \Big( \theta^1_s( E^{i}(s,u,T) )+ \theta^2_s ( F^{i}(s,u,T)) \Big) \mathrm{d}s  \Big),
		\end{align*}
From the definition of $D^{i}(u,T)$ in \eqref{Ei}, we deduce that
		\begin{equation*}
			\tilde M^{i-1}(\i u) = D^{i}(u,T) 	\exp\Big(  \int_0^{t_{i}} \Big( \theta^1_s( E^{i}(s,u,T) )+ \theta^2_s ( F^{i}(s,u,T)) \Big) \mathrm{d}s  \Big), 
		\end{equation*}
		and the expression of $B^1_{i}$ follows from \eqref{B1}. Finally, observe that
		\begin{equation*}
			B_{i}^{2} = e^{-C(t_{i+1} -t_1)} \mathbb{E}_{Q^{S,i}} \Big[\prod_{l=2}^{i+1}  e^{-\beta \Delta t_l |D(t_{l-1})|} \Big].
		\end{equation*}
		The representation of $B^2_i$ follows from the same arguments used above.

\subsection{Proof of Theorem \ref{PDB}}

We recall from \ref{eqdb1} that the price is given by
\begin{align}
	\text{P}^{\text{DB}}=&\sum_{i=1}^N \mathbb{E}_{Q\odot P}\Big[e^{-\int_0^{\bar{t}_i}r(u) \mathrm{d}u}\text{DB}(\bar{t}_i)\Big]\notag\\
	=& \sum_{i=1}^N \mathbb{E}_{Q\odot P}\Big[e^{-\int_0^{\bar{t}_i}r(u) \mathrm{d}u}\max(IS_{\bar{t}_i},Ie^{\delta \bar{t}_i})\mathds{1}_{\{\tau^s\geq \bar{t}_i\}} X_{\alpha, i}\Big].
\end{align}
We focus on the summands of the above equation and distinguish between two cases: when $\bar t_i \le t_1$ and when $t_1 < \bar t_i$. We start with the  description of the second case which is the more technical one. We treat the first case at the end.
	
	For $j\in \{1,\ldots, K-2\}$ and $i$ such that $t_j < \bar t_i \le t_{j+1}$, as well as for $j=K-1$ and  $i$ such that $t_{K-1} < \bar t_i\leq T$,  by similar computations to those done in \eqref{GMAB new1} and \eqref{GMAB new2}, we get 
	\begin{eqnarray}
		&& \mathbb{E}_{Q\odot P}\Big[e^{-\int_0^{\bar{t}_i} r(u)\mathrm{d}u} \mathds{1}_{\{ \tau^s \geq \bar{t}_i\}}\max(I S(\bar{t}_i), Ie^{\delta \bar{t}_i})X_{\alpha, i}\Big]\notag\\
		&&=\mathbb{E}_{Q\odot P}\Big[e^{-\int_0^{\bar{t}_i} r(u)\mathrm{d}u} \mathds{1}_{\{ \tau^s \geq t_{j+1}\}}\max(I S(\bar{t}_i), Ie^{\delta \bar{t}_i})X_{\alpha, i}\Big]\notag\\
		&&= P(i)\mathbb{E}_{Q\odot P}\Big[e^{-\int_0^{\bar{t}_i} r(u)\mathrm{d}u}\max(I S(\bar{t}_i), Ie^{\delta \bar{t}_i})\mathbb{E}_{Q\odot P}\big[\mathds{1}_{\{ \tau^s \geq t_{j+1}\}}|\mathcal{F}_{\max(t_{j+1}, \bar t_i)}^{L^1,L^2}\big]\Big]\notag\\
		&&=P(i)\mathbb{E}_{Q\odot P}\Big[e^{-\int_0^{\bar{t}_i} r(u)\mathrm{d}u} e^{-\int_0^{t_{j+1}} \lambda^s(u) \mathrm{d}u}\max(I S(\bar{t}_i), Ie^{\delta \bar{t}_i})\Big]\notag\\
		&&=P(i)\mathbb{E}_{Q}\Big[e^{-\int_0^{\bar{t}_i} r(u)\mathrm{d}u} e^{-\int_0^{t_{j+1}} \lambda^s(u) \mathrm{d}u}\max(I S(\bar{t}_i), Ie^{\delta \bar{t}_i})\Big]\notag\\
		&&=Ie^{\delta \bar{t}_i} P(i)\mathbb{E}_{Q}\Big[e^{-\int_0^{\bar{t}_i} r(u)\mathrm{d}u} e^{-\int_0^{t_{j+1}} \lambda^s(u) \mathrm{d}u}\Big
		(1+ \Big(\frac{IS_{\bar t_i}}{Ie^{\delta \bar{t}_i}}-1\Big)^+\Big) \Big].    \label{DB1}
	\end{eqnarray}
	The $\bar t_i$-forward measure $Q^{\bar t_i}$ is defined by its Radon-Nikodym density as follows
	\begin{equation}
		\frac{\mathrm{d}Q^{\bar t_i}}{\mathrm{d}Q}= \frac{1}{B(0,\bar t_i)B(\bar t_i)}. \label{second forward measure}
	\end{equation}
	Let $\mathbb{E}_{Q^{\bar t_i}}$ be the expectation with respect to $Q^{\bar t_i}$, then using \eqref{second forward measure} yields
	\begin{eqnarray}
		&& \mathbb{E}_{Q\odot P}\Big[e^{-\int_0^{\bar{t}_i} r(u)\mathrm{d}u} \mathds{1}_{\{ \tau^s \geq \bar{t}_i\}}\max(I S(\bar{t}_i), Ie^{\delta \bar{t}_i})X_{\alpha, i}\Big]\notag\\ 
		&&=Ie^{\delta \bar{t}_i}P(i)B(0,\bar t_i)\Big( \mathbb{E}_{Q^{\bar t_i}}\Big[ e^{-\int_0^{t_{j+1}} \lambda^s(u) \mathrm{d}u}\Big] + \mathbb{E}_{Q^{\bar t_i}}\Big[ e^{-\int_0^{t_{j+1}} \lambda^s(u) \mathrm{d}u}\Big(\frac{S_{\bar t_i}}{e^{\delta \bar{t}_i}}-1\Big)^+ \Big]\Big)\notag\\
		&&= \ Ie^{\delta \bar{t}_i}P(i)B(0,\bar t_i)\big( A_{j,i}^1 + A_{j,i}^2\big). \label{A1A2'}
	\end{eqnarray}  
Using similar arguments as in \eqref{Price-Fourier}
	\begin{equation}
		e^{C(t_{j+1} -t_1)} A_{j,i}^1= \mathbb{E}_{Q^{\bar t_i}}\Big[f(D(t_1), \ldots, D(t_{j}))\Big]=\frac{1}{(2\pi)^{j}}\int_{\mathbb{R}^{j}}\tilde M^{j}_i(\i u)\hat{f}(-u)\mathrm{d}u, \label{Aji1}
	\end{equation}
	where
	\begin{align*}
		f(x_1, \ldots, x_{j})
		=&\prod_{l=2}^{j+1}  e^{-\beta \Delta t_{l} |x_{l-1}|}, 
		\text{ with }\Delta t_l =t_l -t_{l-1}; \,\,\,\,\,	\hat{f}(u_1, \ldots, u_{j})=  \prod_{l=2}^{j+1} \frac{2\beta \Delta  t_{l}}{u_{l-1}^2 + (\beta \Delta  t_{l})^2},\\
			\tilde M^{j}_i(\i u) :=& \mathbb{E}_{Q^{\bar t_i}}\Big[e^{\i u_1 D(t_1) + \ldots + \i u_{j}D(t_{j})}\Big].
	\end{align*}

	Using the dynamics of the bank account in \eqref{Bank ac}, the Radon-Nikodym density in \eqref{second forward measure}  can be written as follows
	\begin{equation} 
		\frac{\mathrm{d}Q^{\bar t_i}}{\mathrm{d}Q}= \exp\Big( -\int_0^{\bar t_i} A(s,\bar t_i)\mathrm{d}s + \int_0^{\bar t_i} \Sigma_1(s,\bar t_i)\mathrm{d}L^1_s -\int_0^{\bar t_i} \Sigma_2(s,\bar t_i)\mathrm{d}L^2_s \Big). \label{second forward measure'}
	\end{equation}
	Therefore, using the representation of $D(t)$ in \eqref{D:appendix}, using \eqref{DB notation} and equation \eqref{cumulant}, we deduce that
	\begin{align*}
		\tilde M^{j}_i(\i u) =& \mathbb{E}_{Q^{\bar t_i}}\Big[e^{\i u_1 D(t_1) + \ldots + \i u_{j}D(t_{j})}\Big]\\
		=& \mathbb{E}_Q\Big[e^{\i u_1 D(t_1) + \ldots + \i u_{j}D(t_{j})-\int_0^{\bar t_i} A(s,\bar t_i)ds + \int_0^{\bar t_i} \Sigma_1(s,\bar t_i)\mathrm{d}L^1_s-\int_0^{\bar t_i} \Sigma_2(s,\bar t_i)\mathrm{d}L^2_s }\Big]\\
		=& \exp\Big( \i\sum_{l=1}^{j} u_{l}\big(-p(t_{l}) - \delta T+ \int_0^T f(0,s)\mathrm{d}s + \int_0^{t_{l}}A(s,T)\mathrm{d}s -\omega(t_{l})\big)-\int_0^{\bar t_i} A(s,\bar t_i)\mathrm{d}s  \Big)\\
		& \times \mathbb{E}_{Q}\Big[
		\exp\Big( \i\sum_{l=1}^{j} \Big(\int_0^{t_{l}} u_{l}(\sigma_2(s)+\Sigma_2(s,T))\mathrm{d}L^2_s - \int_0^{t_{l}}u_{l} \Sigma_1(s,T)\mathrm{d}L^1_s\Big) \\
		& \hspace{1cm}  + \int_0^{\bar t_i} \Sigma_1(s,\bar t_i)\mathrm{d}L^1_s-\int_0^{\bar t_i} \Sigma_2(s,\bar t_i)\mathrm{d}L^2_s \Big)   \Big].\\
		=&\exp\Big( \i\sum_{l=1}^{j} u_{l}\big(-p(t_{l}) - \delta T+ \int_0^T f(0,s)\mathrm{d}s + \int_0^{t_{l}}A(s,T)\mathrm{d}s -\omega(t_{l})\big)-\int_0^{\bar t_i} A(s,\bar t_i)\mathrm{d}s  \Big)\\
		&\times \mathbb{E}_{Q}\Big[\exp\Big(  \int_0^{\bar t_{i}} E_{j,i}(s,u,T) \mathrm{d}L^1_s + \int_0^{\bar t_{i}} F_{j,i}(s,u,T) \mathrm{d}L^2_s \Big) \Big] \\
		=&\exp\Big( \i\sum_{l=1}^{j} u_{l}\big(-p(t_{l}) - \delta T+ \int_0^T f(0,s)\mathrm{d}s + \int_0^{t_{l}}A(s,T)\mathrm{d}s -\omega(t_{l})\big)-\int_0^{\bar t_i} A(s,\bar t_i)\mathrm{d}s  \Big)\\
		&\times \exp\Big(  \int_0^{\bar t_{i}} \Big( \theta^1_s( E_{j,i}(s,u,T) )+ \theta^2_s ( F_{j,i}(s,u,T)) \Big) \mathrm{d}s  \Big)
	\end{align*}
	
	From the definition of $D^{j}(u,T)$ in \eqref{DB notation},  we deduce that
	\begin{equation*}
		\tilde M^{j}_i(\i u) = D^{j}(u,T)  \exp\Big( -\int_0^{\bar t_i} A(s,\bar t_i)\mathrm{d}s + \int_0^{\bar t_{i}} \Big( \theta^1_s( E_{j,i}(s,u,T) )+ \theta^2_s ( F_{j,i}(s,u,T)) \Big) \mathrm{d}s  \Big). 
	\end{equation*}
	By combining \eqref{Aji1} and the definition of $M^{j,i}(u,T)$ in \eqref{DB notation}, we finally obtain
	\[
	A_{j,i}^1=  \frac{e^{-C(t_{j+1} -t_1)}}{(2\pi)^{j}}e^{-\int_0^{\bar t_i} A(s,\bar t_i)ds }\int_{\mathbb{R}^{j}}M^{j,i}(u,T) \mathrm{d}u.
	\]
	
	Let us now compute  $A_{j,i}^2$. For $Y_t=\log(S_t)$, we have
	\[
	\Big(\frac{S_{\bar t_i}}{e^{\delta \bar{t}_i}}-1\Big)^+ = \Big(\exp\big(Y_{\bar t_i} -\delta \bar t_i\big) -1\Big)^+,
	\]
 Note that  
	\begin{equation}
		e^{C(t_{j+1} -t_1)}A_{j,i}^2= \mathbb{E}_{Q^{\bar t_i}}\big[h\big(D(t_1), \ldots, D(t_{j}), Y_{\bar t_i} \big)\big], \label{Aji2}
	\end{equation}
	with $h(x_1, \ldots, x_{j+1}):=f(x_1,\dots,x_{j})  (e^{x_{j+1} -\delta\bar t_i }-1)^+ ,$  where $f$ given in \eqref{Aji1}.
	To ensure integrability,  we define $H(x_1,\ldots, x_{j+1}):=h(x_1, \ldots, x_{j+1}) e^{-rx_{j+1}}$, for some $1<r<2$, and 
	\[
	H_{j+1} (x_{j+1}) :=  (e^{x_{j+1} -\delta\bar t_i }-1)^+ e^{-r x_{j+1}}. 
	\]
	Therefore, $H_{j+1} \in L^1(\mathbb{R})$, and $H\in L^1(\mathbb{R}^{j+1})$.
	By simple  integration we get that
	for all $ y \in \mathbb{R}$
	\[
	\hat{H}_{j+1}(y)= \frac{\exp\big(\delta\bar t_i(\i y-r)\big)}{(\i y -r+1)(\i y-r)}. 
	\] 
	We have $|\hat{H}_{j+1}(y)|_{\mathbb{C}}= e^{-r\delta\bar t_i}(((1-r)^2+ y^2)(r^2+y^2))^{-1/2}$, so,
	$\hat{H}_{j+1}\in L^1(\mathbb{R})$. Consequently, using the above equality and \eqref{eq:fhat}, we obtain $\hat{H}\in L^1(\mathbb{R}^{j+1})$, and
	\begin{equation}
		\hat{H}(y_1, \ldots, y_{j+1})= \frac{\exp\big(\delta\bar t_i(\i y_{j+1} -r)\big)}{(\i y_{j+1} -r+1)(\i y_{j+1}-r)} \prod_{l=2}^{j+1} \frac{2\beta \Delta  t_{l}}{y_{l-1}^2 + (\beta \Delta  t_{l})^2}. \label{DB fourier2}
	\end{equation}
	Since $H, \hat{H}\in L^1(\mathbb{R}^{j+1})$, we deduce  from \cite[Theorem 3.2]{Eberlein2010} that
	\begin{equation}
		\mathbb{E}_{Q^{\bar t_i}}\big[h\big(D(t_1), \ldots, D(t_{j}), Y_{\bar t_i}\big)\big] = \frac{1}{(2\pi)^{j+1}}\int_{\mathbb{R}^{j+1}}\tilde N_i^{j+1}(R+\i u)\hat{h}(\i R-u)\mathrm{d}u, \label{DB fourier3}
	\end{equation}
	where $R=(0, \ldots, 0, r)\in\mathbb{R}^{j+1}$, $1<r<2$, and  
	\begin{equation}
		\tilde N_i^{j+1}(R+\i u):= \mathbb{E}_{Q^{\bar t_i}}\big[e^{\i u_1 D(t_1) + \ldots + \i u_{j}D(t_{j})+ (\i u_{j+1} +r)Y_{\bar t_i}}\big]. \label{DB tildeN}
	\end{equation}
	Using once more $Y_{\bar t_i}=\log(S_{\bar t_i})$, we deduce  from \eqref{Stock price} that 
	\[
	Y_{\bar t_i}=\int_0^{\bar t_i} r(s)\mathrm{d}s+\int_0^{\bar t_i} \sigma_2(s)\mathrm{d}L^2_s -\omega(\bar t_i),
	\]
	Thanks to \eqref{interest}, we obtain
	\begin{align*}
		Y_{\bar t_i}=& \int_0^{\bar t_i} f(0,s)\mathrm{d}s + \int_0^{\bar t_i}A(s,\bar t_i)\mathrm{d}s - \int_0^{\bar t_i}\Sigma_1(s,\bar t_i)\mathrm{d}L^1_s + \int_0^{\bar t_i}\Sigma_2(s,\bar t_i)\mathrm{d}L^2_s +\int_0^{\bar t_i} \sigma_2(s)\mathrm{d}L^2_s -\omega(\bar t_i)\\
		=& \int_0^{\bar t_i} f(0,s)\mathrm{d}s + \int_0^{\bar t_i}A(s,\bar t_i)\mathrm{d}s - \int_0^{\bar t_i}\Sigma_1(s,\bar t_i)\mathrm{d}L^1_s+ \int_0^{\bar t_i}(\sigma_2(s)+\Sigma_2(s,\bar t_i))\mathrm{d}L^2_s-\omega(\bar t_i).
	\end{align*}
	Substituting this into \eqref{DB tildeN}, we get
	\begin{align*}
		&\tilde N_i^{j+1}(R+\i u)\\
		=& \mathbb{E}_Q\Big[e^{\i u_1 D(t_1) + \ldots + \i u_{j}D(t_{j}) +(\i u_{j+1} +r)Y_{\bar t_i}-\int_0^{\bar t_i} A(s,\bar t_i)ds + \int_0^{\bar t_i} \Sigma_1(s,\bar t_i)\mathrm{d}L^1_s -\int_0^{\bar t_i} \Sigma_2(s,\bar t_i)\mathrm{d}L^2_s}\Big]\\
		=& \exp\Big( \i\sum_{l=1}^{j} u_{l}\big(-p(t_{l}) - \delta T+ \int_0^T f(0,s)\mathrm{d}s + \int_0^{t_{l}}A(s,T)\mathrm{d}s -\omega(t_{l})\big)-\int_0^{\bar t_i} A(s,\bar t_i)\mathrm{d}s   \Big)\\
		&\times \exp\Big((\i u_{j+1} +r) \big( \int_0^{\bar t_i} f(0,s)\mathrm{d}s + \int_0^{\bar t_i}A(s,\bar t_i)\mathrm{d}s-\omega(\bar t_i)\big) \Big)\\
		&  \times \mathbb{E}_{Q}\Big[
		\exp\Big( \i\sum_{l=1}^{j} \Big(\int_0^{t_{l}} u_{l}(\sigma_2(s)+\Sigma_2(s,T))\mathrm{d}L^2_s - \int_0^{t_{l}}u_{l} \Sigma_1(s,T)\mathrm{d}L^1_s\Big)\\ 
		& +\;(\i u_{j+1} +r) \Big(\int_0^{\bar t_{i}} (\sigma_2(s)+\Sigma_2(s, \bar t_i))\mathrm{d}L^2_s 
		- \int_0^{\bar t_{i}}\Sigma_1(s,\bar t_i)\mathrm{d}L^1_s \Big)\\ 
		& +\int_0^{\bar t_i} \Sigma_1(s,\bar t_i)\mathrm{d}L^1_s -\int_0^{\bar t_i} \Sigma_2(s,\bar t_i)\mathrm{d}L^2_s \Big)   \Big].
	\end{align*}
	Using again the notations in \eqref{DB notation}, the last quantity  can be rewritten as follows
	\begin{align*}
		\tilde N^{j+1}_i(R+\i u)
		= &\tilde D^{j,i}(u -\i R ,T)e^{-\int_0^{\bar t_i} A(s,\bar t_i)\mathrm{d}s}\\
		&\times
		\; \mathbb{E}_Q\Big[\exp\Big( \int_0^{\bar t_i} \tilde E_{j,i}(s,u-\i R,T) \mathrm{d}L^1_s + \int_0^{\bar t_i} \tilde F_{j,i}(s,u-\i R,T) \mathrm{d}L^2_s \Big)\Big].
	\end{align*}
	Since $1<r <2$, using \eqref{AssM2}, \eqref{AssM1}  and \eqref{AssM22}, we have  
	$$|r \sigma_2(s)+(r-1)\Sigma_2(s, \bar t_i)|\leq 2\sigma_2(s) +(r-1)|\Sigma_2(s, \bar t_i)|\leq 2\frac{M_2}{3} +\frac{M_2}{3} \leq M_2$$
	and 
	$$|(1-r)\Sigma_1(s,\bar t_i)| \leq M_1.$$
	Therefore, the above expectation exists.
	By the independence of $L^1$ and $L^2$ and \eqref{cumulant}, we deduce that \small
	\begin{eqnarray}
		\tilde N_i^{j+1}(R+\i u)
		&=& \tilde D^{j,i}(u-\i R  ,T) e^{-\int_0^{\bar t_i} A(s,\bar t_i)\mathrm{d}s}\notag\\
		&&\times \exp\Big( \int_0^{\bar t_i} \theta^1_s(\tilde E_{j,i}(s,u-\i R,T)) \mathrm{d}s + \int_0^{\bar t_i} \theta^2_s(\tilde{F}_{j,i}(s,u-\i R,T))    \mathrm{d}s \Big). \label{DB N'}
	\end{eqnarray}\normalsize
	On the other hand, for any $u\in \mathbb{R}^{j+1}$, it holds
	\[
	\hat{H}(u)=\int_{\mathbb{R}^{j+1}}e^{\i \langle u, x \rangle } e^{- \langle R , x \rangle} h(x)dx = \hat{h}(u+\i R). 
	\]
	Therefore, we get
	\begin{equation}
		\hat{h}(\i R - u) =\hat{H}(-u) 
		= \frac{\exp\big(-\delta \bar t_i (\i u_{j+1} +r)\big)}{(\i u_{j+1} +r-1)(\i u_{j+1}+r)} \prod_{l=2}^{j+1} \frac{2\beta \Delta  t_{l}}{u_{l-1}^2 + (\beta \Delta  t_{l})^2}. \label{DB fourier4}
	\end{equation}
	Finally, substituting \eqref{DB N'} and \eqref{DB fourier4} into \eqref{DB fourier3}, and using the definition of $ N^{j,i}(u,T)$ in \eqref{DB notation}, we obtain
	\[
	A_{j,i}^2= \frac{e^{-C(t_{j+1} -t_1)}}{(2\pi)^{j+1}}e^{-\int_0^{\bar t_i} A(s,\bar t_i)ds }\int_{\mathbb{R}^{j+1}}N^{j,i}(u,T) \mathrm{d}u.
	\]
	It remains to check the case $\bar t_i \le t_1$ for $i\in \{1, \ldots, N\}$. Note that for $\bar t_i \le t_1$, $\mathds{1}_{ \{\tau^s \geq \bar{t}_i\}}=\mathds{1}_{ \{\tau^s \geq t_1\}}=1$. Thus
	\begin{eqnarray*}
		\mathbb{E}_{Q\odot P}\Big[e^{-\int_0^{\bar t_i} r(u)\mathrm{d}u} {\rm DB}(\bar t_i)\Big]
		&=& \mathbb{E}_{Q\odot P}\Big[e^{-\int_0^{\bar{t}_i} r(u)\mathrm{d}u} \mathds{1}_{ \{\tau^s \geq \bar{t}_i\}}\max(I S(\bar{t}_i), Ie^{\delta \bar{t}_i})X_{\alpha,i}\Big]\\
		&=&\mathbb{E}_{Q\odot P}\Big[e^{-\int_0^{\bar{t}_i} r(u)\mathrm{d}u} \max(I S(\bar{t}_i), Ie^{\delta \bar{t}_i})X_{\alpha,i}\Big].
	\end{eqnarray*}
 By similar arguments as above, we get
	\begin{eqnarray*}
		\mathbb{E}_{Q\odot P}\Big[e^{-\int_0^{\bar{t}_i} r(u)\mathrm{d}u} \max(I S(\bar{t}_i), Ie^{\delta \bar{t}_i})X_{\alpha,i}\Big] =P(i) \mathbb{E}_{Q}\Big[e^{-\int_0^{\bar{t}_i} r(u)\mathrm{d}u} \max(I S(\bar{t}_i), Ie^{\delta \bar{t}_i})\Big].
	\end{eqnarray*}
	By the same arguments that have been used above, we also get 
	\begin{align*}
		\mathbb{E}_{Q}\Big[e^{-\int_0^{\bar{t}_i} r(u)\mathrm{d}u} \max(I S(\bar{t}_i), Ie^{\delta \bar{t}_i})\Big]
		 =& B(0, \bar t_i) \mathbb{E}_{Q^{\bar t_i}}\Big[\max(I S(\bar{t}_i), Ie^{\delta \bar{t}_i})\Big]\\
		=& B(0, \bar t_i)Ie^{\delta \bar{t}_i} \mathbb{E}_{Q^{\bar t_i}}\Big[(1+ \Big(\frac{IS_{\bar t_i}}{Ie^{\delta \bar{t}_i}}-1\Big)^+)\Big]\\
		=& B(0, \bar t_i)Ie^{\delta \bar{t}_i} + B(0, \bar t_i)Ie^{\delta \bar{t}_i} \mathbb{E}_{Q^{\bar t_i}}\Big[\Big(\exp\big(Y_{\bar t_i} -\delta \bar t_i\big) -1\Big)^+\Big]\\
		 =& B(0, \bar t_i)Ie^{\delta \bar{t}_i} + B(0, \bar t_i)Ie^{\delta \bar{t}_i} \mathbb{E}_{Q^{\bar t_i}}\big[h_1(Y_{\bar t_i})\big],
	\end{align*}
	with $h_1(x):=(e^{x -\delta \bar t_i }-1)^+$. For some $1<r<2$, let 
	$H_1 (x) :=  (e^{x -\delta \bar t_i }-1)^+ e^{-r x}.$ Thanks to \cite[Theorem 3.2]{Eberlein2010}, we have  
	\begin{equation}
		\mathbb{E}_{Q^{\bar t_i}}\big[h_1\big(Y_{\bar t_i}\big)\big] = \frac{1}{2\pi}\int_{\mathbb{R}} \tilde N_i(r+\i u)\hat{h}_1(\i r-u)\mathrm{d}u, \label{DB fourier5}
	\end{equation}
	where
	\[
	\tilde N_i(r+\i u):=\mathbb{E}_{Q^{\bar t_i}}\big[e^{(r+ \i u )Y_{\bar t_i}}\big].
	\]
	Once again similar arguments as before, and the definitions in \eqref{DB notation}, we have
	\begin{align}
		&\tilde N_i(r+\i u)\notag\\
		=& \exp\Big((r +\i u)w_{\bar t_i} -\int_0^{\bar t_i} A(s,\bar t_i)\mathrm{d}s \Big)\; \mathbb{E}_Q\Big[\exp\Big(\int_0^{\bar t_i}E_i(s,u)\mathrm{d}L^1_s + \int_0^{\bar t_i}F_i(s,u)\mathrm{d}L^2_s\Big)\Big]\notag\\
		=& \exp\Big((r +\i u)w_{\bar t_i} -\int_0^{\bar t_i} A(s,\bar t_i)\mathrm{d}s \Big)\;\exp\Big( \int_0^{\bar t_i}\theta^1_s(E_i(s,u))\mathrm{d}s + \int_0^{\bar t_i}\theta^2_s(F_i(s,u))\mathrm{d}s\Big). \label{DB Ni}
	\end{align}
	On the other hand, we have
	\begin{equation}
		\hat{h}_1(\i r - u) =\hat{H}_1(-u) 
		= \frac{\exp\big(-\delta \bar t_i(\i u +r)\big)}{(\i u +r-1)(\i u+r)}. \label{DB fourier6}
	\end{equation}
	Substituting \eqref{DB Ni} and \eqref{DB fourier6} into \eqref{DB fourier5},  we obtain
	\[
	\mathbb{E}_{Q^{\bar t_i}}\big[h_1\big(Y_{\bar t_i}\big)\big]=\frac{1}{2\pi}e^{-\int_0^{\bar t_i} A(s,\bar t_i)\mathrm{d}s}\int_{\mathbb{R}} N^i(u)\mathrm{d}u=A_{0,i},
	\]
	Thus, 
	\begin{equation*}
		\mathbb{E}_{Q}\Big[e^{-\int_0^{\bar{t}_i} r(u)\mathrm{d}u} \max(I S(\bar{t}_i), Ie^{\delta \bar{t}_i})\Big] =B(0, \bar t_i)Ie^{\delta \bar{t}_i} + B(0, \bar t_i)Ie^{\delta \bar{t}_i}A_{0,i}.
	\end{equation*} 
	This implies that 
	\[
	\mathbb{E}_{Q\odot P}\Big[e^{-\int_0^{\bar t_i} r(u)\mathrm{d}u} {\rm DB}(\bar t_i)\Big]=P(i)( Ie^{\delta \bar{t}_i}B(0, \bar t_i) + Ie^{\delta \bar{t}_i}B(0, \bar t_i) A_{0,i}).
	\]

\section{Numerical implementation}
In this section, we numerically compute the price of the life insurance contract discussed in our paper. Following the approach of \cite{Eberlein-R}, we assume that the function $\sigma_2(\cdot)$ in \eqref{Stock price} is a positive constant, i.e. $\sigma_2(s)=\sigma_2$. The L\'evy processes $L^1$ and $L^2$ are modeled as Normal Inverse Gaussian (NIG) processes.
The cumulant function for these NIG processes is given by the following expression:
\[
\theta^i(z)= \delta_i\Big(\sqrt{\alpha_i^2 -\beta_i^2} - \sqrt{\alpha_i^2-(\beta_i +z)^2}\Big), \quad -\alpha_i -\beta_i <Re(z) < \alpha_i -\beta_i,
\] 
for $i\in \{1,2\}$. At this step, we need to determine the constants $\alpha_1, \beta_1, \delta_1, \alpha_2, \beta_2,\delta_2$, the constant $\sigma_2$, and the constants $a,b$ given in \eqref{sigma un} and \eqref{beta un}, respectively. These constants are taken from the calibration done in \cite{Eberlein-R} and are provided in the following table. For the mortality model, to simplify notation, we will use $\phi^i=(\lambda_i, \mu_i, \sigma_i, \epsilon_i, \kappa_i)$ as defined in Theorem \ref{joint density}. All the parameters for the contract are also summarised in the table below

\begin{table} [H]
	\centering
	\caption{Parameters of the VA contract}
	\small\begin{tabular}{ | c c | c c c | c c| c c | c c c| }
		\hline
		\multicolumn{2}{c}{Insurance policy}   & \multicolumn{5}{c}{Financial market model}  & \multicolumn{2}{c}{Surrender model} & \multicolumn{3}{c}{Mortality model} \\ 
		\hline \hline
		& & & $L_t^1$ &  $L_t^2$ & & & & & & & \\
		$I$ & 100 & $\alpha_i$ & $3.12$ & $3.31$ & & &  $\beta$ & $0.02$ & $\lambda_i$ & $0.3$ & $0.3$\\
		$T$ & 3, 10 years & $\beta_i$ & $1.87$ &  $-1.43$ & & & $C$ & $0.005$ & $\mu_i$ & $0.07$ & $0.05$\\
		$\delta$ & 0.02 p.a. &$\delta_i$& $9.24$ & $6.21$ & & & & & $\sigma_i$ & $0.005$ & $0.002$\\
		$\tilde{P}(t_l)$ &$0.95+0.05t_l/T$ & & &  & $a$ &  $0.00258$ & & & $\epsilon_i$ & $1$ & $1$\\
		$t_l -t_{l-1}$ & 1 year &  &  &  & $b$ & $0.00143$ &  & & $\kappa_i$ & $0.5$ & $0.5$ \\
		$\bar{t}_i-\bar{t}_{i-1}$ & 6 months &  & &  & $\sigma_2$ & $0.1559$ & & & & &\\ 
		\hline
	\end{tabular}
	\caption*{}
	\label {table: Params_VA_contract}
\end{table}

\subsection{Implementation}

The high peaks around the origin in both functions $M$ and $N$,
as shown in Figures \ref{Real_imag_MM} and \ref{Real_imag_N}, create significant challenges for standard Monte Carlo integration. These sharp variations result in high variance, making it difficult to achieve accurate estimates without requiring an impractically large number of samples. This issue is well-documented in the literature, where integrals involving localized peaks or regions with rapid changes are prone to inefficiencies in Monte Carlo methods due to their disproportionate impact on the estimator's overall variance (see \cite{Caflisch98, Liu01})

To address these challenges, we employ Monte Carlo integration with importance sampling, which enhances the efficiency of evaluating multidimensional integrals. While quadrature methods are suitable for lower-dimensional cases (e.g., when $T\leq 4$), they become impractical for higher-dimensional problems due to the exponential growth in computational cost. Monte Carlo methods, particularly with importance sampling, provide a practical solution by reducing variance through the strategic sampling of more ``important" regions of the integrand (see \cite{RoCa10}). Generally, the Monte Carlo estimate of a high-dimensional integral is given by 
\begin{align*}
	I = \int_{\Omega}^{}f(u)\mathrm{d}u
\end{align*}
over the domain $\Omega\subseteq \mathbb{R}^d$ and can be approximated using Importance Sampling (see \cite{ElMar22}) by 
\begin{align*}
	I = \frac{1}{N}\sum_{i=1}^{N}\frac{f(x_i)}{p(x_i)},
\end{align*}
where $x_i$ is drawn randomly in $\Omega$ with probability density function $p(x)$, and $N$ is the total number of samples. Following \cite{VA1}, we obtain the accuracy of the Monte Carlo integral, $I_{MC}$, in terms of absolute value bias with respect to the quadrature integral, $I_Q$, i.e.
\begin{align}\label{Bias}
	100\times\frac{\left|I_{MC}-I_Q\right|}{I_{MC}}, 
\end{align} 
and the percentage error 
\begin{align}\label{Std_error}
	100\times\frac{1}{I_{MC}}\sqrt{\frac{\sum_{i=1}^{N}(I_i-I_{MC})^2}{N(N-1)}}.
\end{align}
The next two graphs show the behaviour of the integrands, $M$ and $N$ defined in \eqref{E}.

\begin{figure}[H]
	\centering
	\includegraphics[width=16cm,height=6cm]{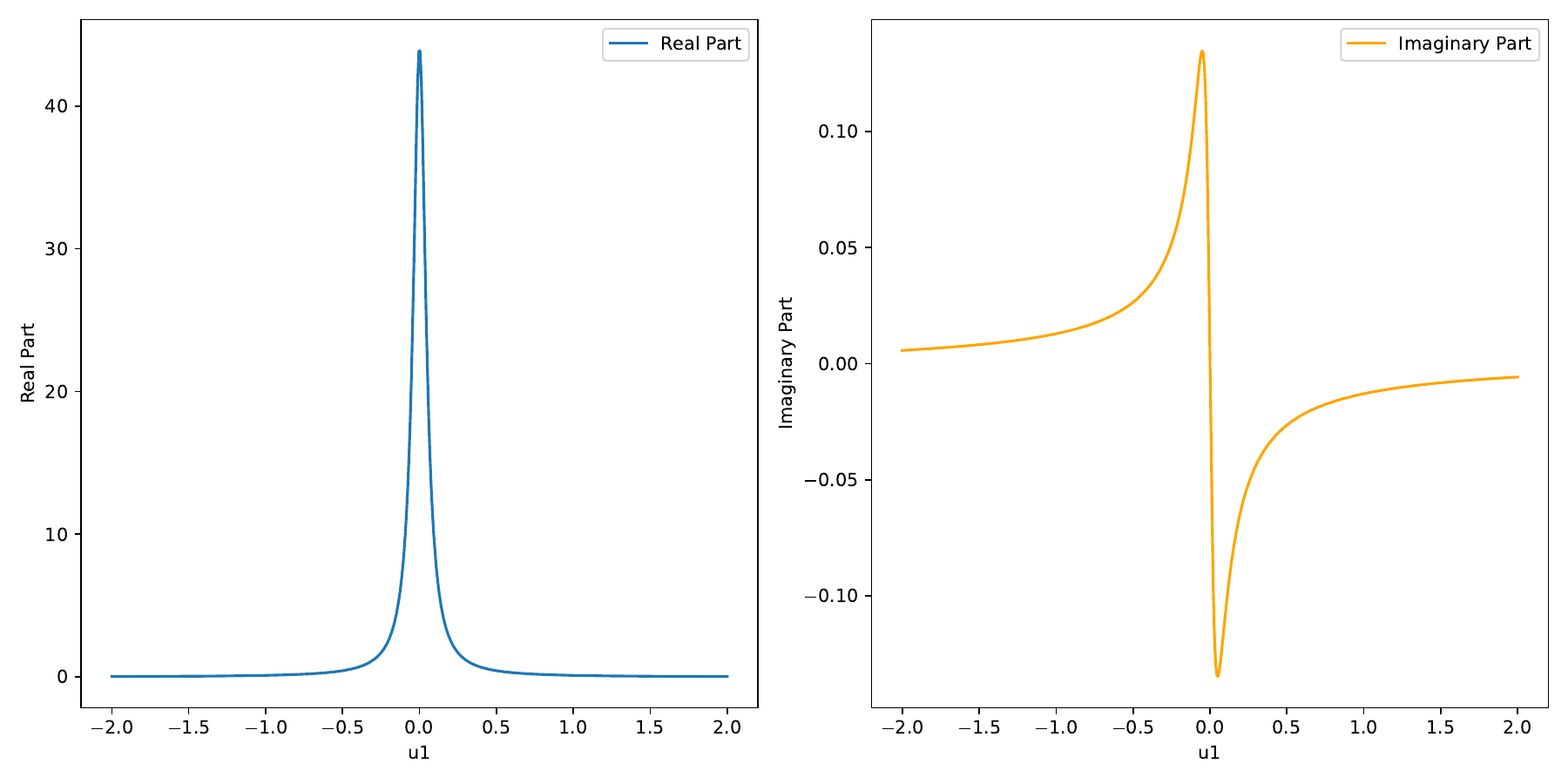}
	\caption{$M(u,T), T = 3\text{ years }(K=2)$. Left panel : real part of $M(u,T)$. Right panel: imaginary part of $M(u,T)$. Parameters: table \ref{table: Params_VA_contract}.}
	\label{Real_imag_MM}
\end{figure}

\begin{figure}[H]
	\centering
	\includegraphics[width=16cm,height=6cm]{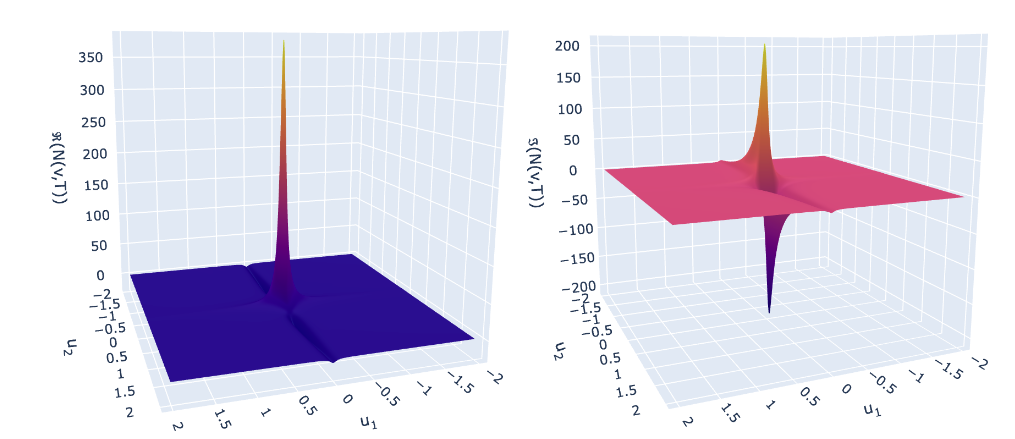}
	\caption{$N(v,T), T = 3\text{ years }(K=2)$. Left panel : real part of $N(v,T)$. Right panel: imaginary part of $N(v,T)$. Parameters: table \ref{table: Params_VA_contract}.}
	\label{Real_imag_N}
\end{figure}

\subsection{Benchmarking and testing}

To provide a reliable estimate and benchmark for the value of the integrals in $A_1$, $A_2$, $B_i^1, B_i^2$, $A^1_{j,i}$, and $A^2_{j,i}$ obtained through the Monte Carlo procedure\footnote{The integration was conducted on an Apple Mac Mini equipped with an M2 Pro chip, 32 GB of unified memory, a 12-core CPU, and a 19-core GPU, utilizing parallel computing to enhance performance.}, we implement the integration in lower dimension (for $T=3$ in the case of integrals related to the GMAB price and DB price, and $T=4$ for integrals related to SB price) using adaptive quadrature methods in Python (see https://scipy.org/citing-scipy/) and compare the results in Table \ref{table: Benchmark_MC_vs_quad}.

We implement an example of a contract with maturity $T = 3$, (i.e. $K = 2$) and annually spaced termination dates $t_i$. This implies that $A_1$ and $A_2$ are, respectively one- and two-dimensional integrals, $B_1^1 = 1$ and $B_1^2$ is one-dimensional integral. Moreover, $A_{0,i}$ and $A_{j,i}^1$ are one-dimensional integrals, and $A_{j,i}^2$ is a two-dimensional integral. All numerical results from the benchmark shown in Table \ref{Benchmark_table}, confirm that the biases and standard errors are below $0.5\%$.

\begin{table}[H]
	\caption{Benchmarking Monte Carlo Integration with IS (the case for GMBA/SB/DB)}
	\label{Benchmark_table}
	\centering
	
	\setlength{\tabcolsep}{15pt} 
	\small\begin{tabular}{@{} c c c c c c c @{} >{\kern\tabcolsep}l @{}}    
		\toprule
		&  & $N$ & \multicolumn{1}{c}{Quadrature} & \multicolumn{3}{c}{Monte Carlo integration} &  \\
		\cmidrule(lr){3-3} \cmidrule(lr){4-4} \cmidrule(lr){5-8}
		$\text{Benefit}$ & $T =3,~~ K=2$ & value & value & Value & Bias(\%) & Std. Error (\%) &  \\
		\midrule
		GMAB & $A_1$ & $75$ million & 0.9907 & 0.9906 & 0.0029 & 0.002 & \\ 
		& $A_2$ &  & $0.1345$ & 0.1343& 0.1257 & 0.0935 &\\ 
		& 	\vtop{\hbox{\strut Processor}\hbox{\strut run time}} & & 6515.36 & 1415.6 &  &  &\\
		\midrule
		SB & $B_2^1$ & $50$ million & 1 & 1 & 0 & - &  \\ 
		& $B_2^2$ &  & 0.991103 & 0.99111 & 0.00062 & 0.0024 &\\ 
		& \vtop{\hbox{\strut Processor}\hbox{\strut run time}} & &  0.3513 & 1415.3997 &  &  &\\
		\midrule
		DB & $A^1_{2,3}$ & $50$ million & 0.9481 & 0.9482 & 0.0092  & 0.0054 & \\ 
		& $A^2_{2,3}$ &  & 0.1089 & 0.1087& 0.1612 & 0.3456 &\\ 
		& \vtop{\hbox{\strut Processor}\hbox{\strut run time}} & & 692.9797 & 1550.6255 &  &  &\\
		\bottomrule
	\end{tabular}
	\caption*{\small{Parameters: table \ref{table: Params_VA_contract}. `integrate' module in Scientific Python (Scipy): built-in functions \textit{quad}, \textit{dblquad}, \textit{tplquad}, \textit{nquad}. Bias\eqref{Bias} / standard error \eqref{Std_error} expressed as percentage of the actual value. Processor (GPU + CPU) overall run time  expressed in seconds. Sample size: $N$.}}
	\label{table: Benchmark_MC_vs_quad}
\end{table}

To  ensure that the MC integrals in the case of the SB price are reliable in the multidimension case, we compute the relevant integrals in the $SB$ price when $T=4$ (i.e $K=3$) and compare the results with the quadrature method in the table below.

\begin{table}[H]
	\caption{Benchmarking Monte Carlo Integration with IS for SB when $T = 4$}
	\centering
	
	\setlength{\tabcolsep}{15pt} 
	\small\begin{tabular}{@{} c c c c c c c @{} >{\kern\tabcolsep}l @{}}    
		\toprule
		&  & $N$ & \multicolumn{1}{c}{Quadrature} & \multicolumn{3}{c}{Monte Carlo integration} &  \\
		\cmidrule(lr){3-3} \cmidrule(lr){4-4} \cmidrule(lr){5-8}
		$\text{Benefit}$ & $T = 4, ~K-1=2$ & value & value & Value & Bias(\%) & Std. Error (\%) &  \\
		\midrule
		SB & $B_2^1$ & $25$ million & 0.99098 & 0.99103 & 0.0048 & 0.0034 & \\ 
		& $B_2^2$ &  & 0.98064 & 0.98064& 0.0002 & 0.0049 &\\ 
		& \vtop{\hbox{\strut Processor}\hbox{\strut run time}} & & $814.8838$ & 5584.1351 &  &  &\\
		\bottomrule
	\end{tabular}
	\caption*{\small{Parameters: table \ref{table: Params_VA_contract}. `integrate' module in Scientific Python (Scipy): built-in functions \textit{quad}, \textit{dblquad}, \textit{tplquad}, \textit{nquad}. Bias\eqref{Bias} / standard error \eqref{Std_error} expressed as percentage of the actual value. Processor (GPU + CPU) overall run time  expressed in seconds. Sample size : $N$.}}
	\label{table: Benchmark_SB_vs_quad_t_is_4}
\end{table}

The following table gives the prices of the life insurance policies for different maturity times.

\begin{table}[H]
	\centering
	\caption{VA by Monte Carlo integration with IS}
	\small\begin{tabular}{ | c | c  c c c | }
		\hline
		$T$ & GMAB &  DB & SB & VA\\
		\hline \hline
		3 years & $13.125$ & $87.2772$ &$0.2102$  & $100.6124$\\
		(Mean std error\%) & $0.0477$ & $0.1856$ & $0.0024$ &   \\
		Processor run time & $1429.5$ & $8426.576$ & $1418.17$ &  \\
		$N$ & $75$ million & $50$ million & $50$ million &  \\
		\midrule
		10 years &$0.0429$ & $108.9348$ & $0.6888$ &  $109.628$\\
		(Mean std error\%) & $0.0212$ & $0.0626$ & $0.0072$ &   \\
		Processor run time\tablefootnote{
			As the dimensions of integration and the number of integrals to estimate increase for both the SB and DB cases, the computational load on the processor and memory significantly intensifies. To mitigate this issue and accelerate the Monte Carlo integration process, the sample sizes for SB and DB were strategically reduced. Specifically, for $T=10$, sample sizes of $N=5\times10^{7}, N=2.5\times10^{7}, N=1\times 10^{7}$ were employed to estimate the integrals under GMAB, SB, and DB, respectively. } & $1233.28$ & $10832.67$ & $81535.25$ &  \\
		$N$ & $50$ million & $15$ million & $25$ million &  \\
		\hline
	\end{tabular}
	\caption*{Parameters: table \ref{table: Params_VA_contract}. Mean Standard error accross all integration procedure.\\ Overall CPU+GPU processing time in seconds. Sample size: $N$.}
	\label {table: VA_integs_MC}
\end{table}

\subsection{Numerical interpretation}
In this section, we provide some numerical interpretation. 
\begin{figure}[H]
	\centering
	\includegraphics[width=16cm,height=10cm]{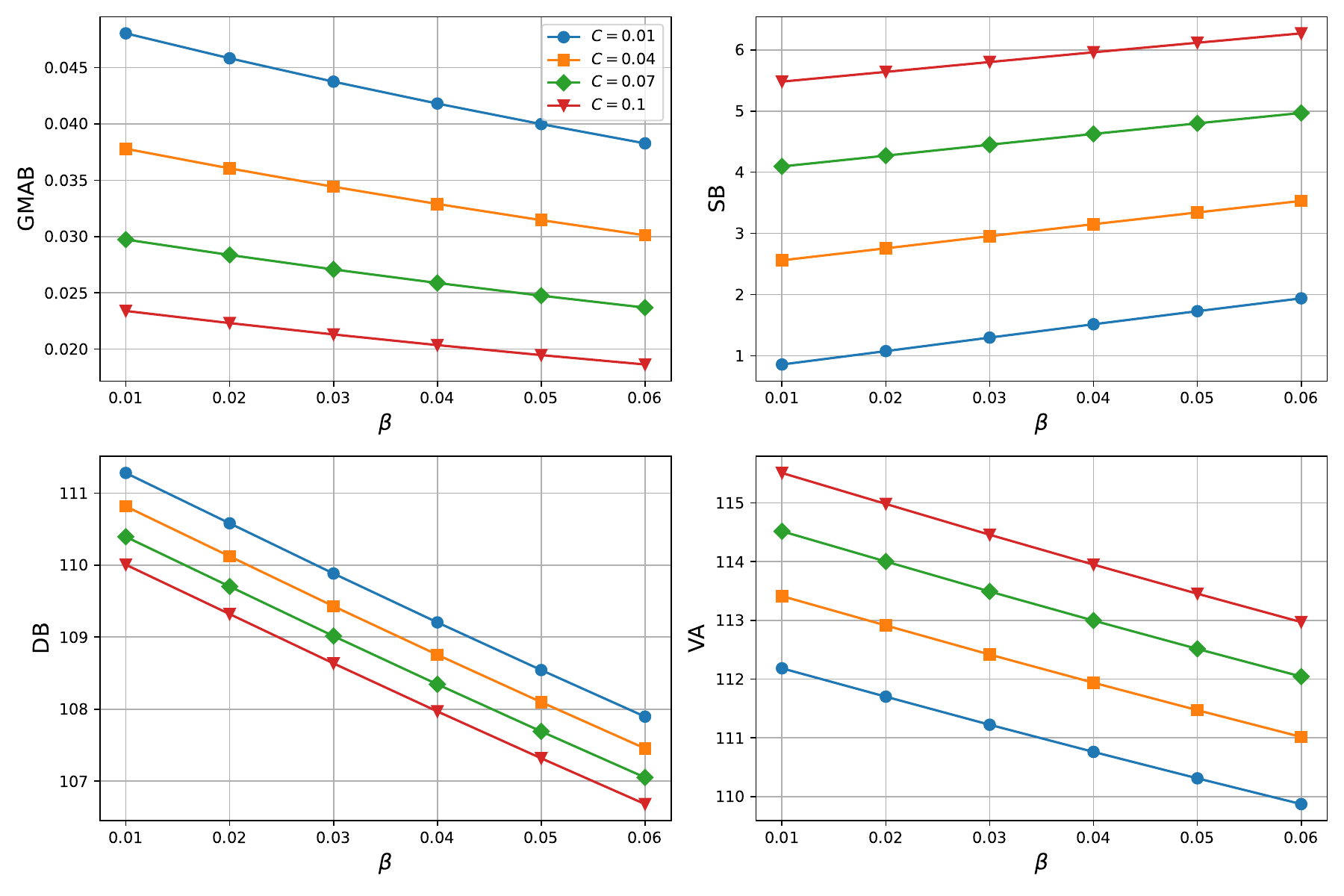}
	\caption{Sensitivity analysis : the parameters $(\beta,C)$ Top panels: left-hand-side—GMAB; right-hand-side: SB. Bottom panels: left-hand-side—DB; right-hand-side—VA. Maturity: T = 10 years.}.
	\label{fig:Sensitivity_P_to_beta_C}
\end{figure}
Figure \ref{fig:Sensitivity_P_to_beta_C} illustrates the impact of $\beta$ and $C$ on the values of the SB, GMAB, DB, as well as the overall value of the life policy. A higher value of $\beta$ corresponds to a greater intensity of surrender, which increases the value of the SB. As $\beta$ rises, the probability of surrender also increases, leading to a decrease in the values of the GMAB and DB, as shown in the graph. Additionally, the policyholder's overall value decreases as $\beta$ increases. A similar observation holds for changes in $C$: as $C$ increases, the values of SB, GMAB, and DB are affected. However, in this case, the policy's value increases, which makes sense as a higher $C$ discourages policyholders from surrendering the contract for non-essential personal reasons
\begin{figure}[H]
	\centering
	\includegraphics[width=16cm,height=10cm]{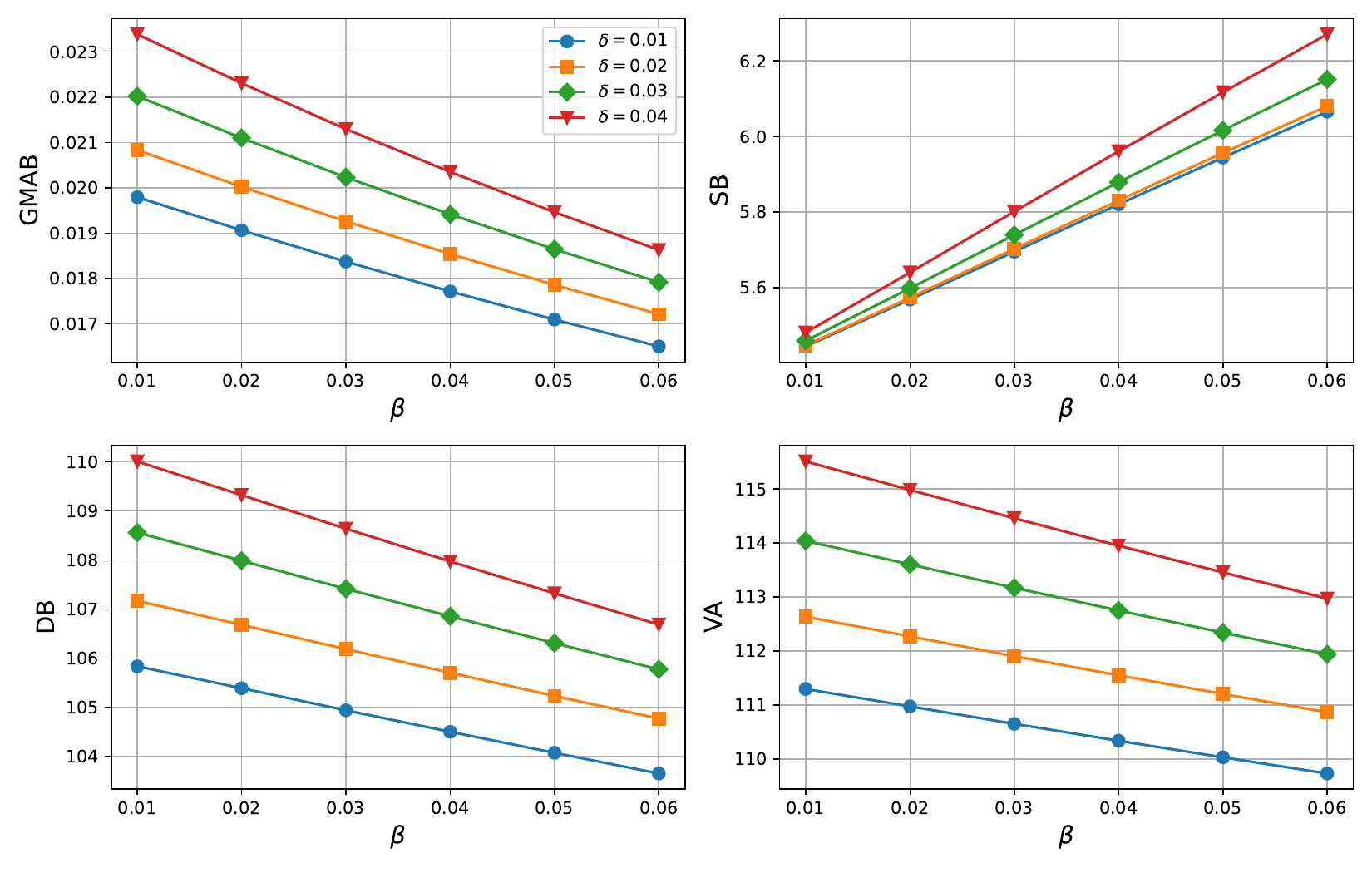}
	\caption{Sensitivity analysis : the parameters $(\beta,\delta)$ Top panels: left-hand-side—GMAB; right-hand-side: SB. Bottom panels: left-hand-side—DB; right-hand-side—VA. Maturity: T = 10 years.}
	\label{fig:Sensitivity_P_to_beta_delta}
\end{figure}

Figure \ref{fig:Sensitivity_P_to_beta_delta} illustrates the impact of $\beta$ and $\delta$ on the values of the SB, GMAB, DB, as well as the overall value of the life policy. As $\delta$ increases, the values of the SB, GMAB, and DB also rise, leading to an increase in the overall value of the policy.

\begin{figure}[H]
	\centering
	\includegraphics[width=14.5cm,height=5cm]{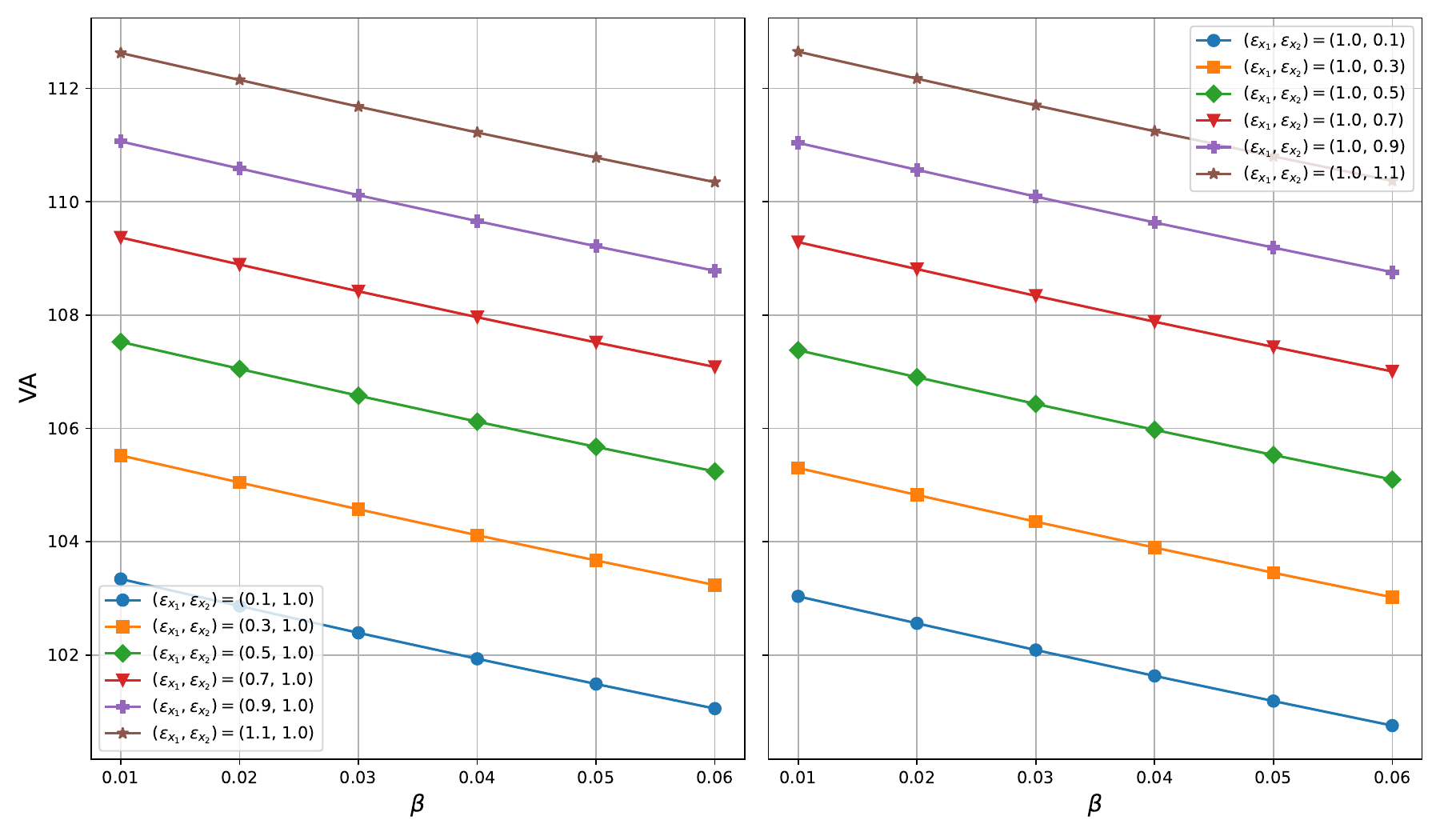}
	\caption{Sensitivity analysis : the mortality parameter $(\epsilon_{x1},\epsilon_{x2})$  Maturity: T = 10 years.}
	\label{fig:Sensitivity_P_to_mort_epsilon}
\end{figure}

Figure \ref{fig:Sensitivity_P_to_mort_epsilon} illustrates the impact of the short-term dependence structure of broken heart syndrome on the life policy's value. As the mortality jump $\epsilon_x$ increases, the value of the contract rises. This is consistent with pricing intuition. A higher $\epsilon_x$ indicates a greater likelihood that one spouse will die shortly after the other, leading to either a large lump-sum payment or two consecutive claims in quick succession. This increase in claims within a short time frame raises the financial risk for the insurer. To mitigate this, the value of the policy increases, compensating for the added risk.

\begin{figure}[H]
	\centering
	\includegraphics[width=14.5cm,height=6cm]{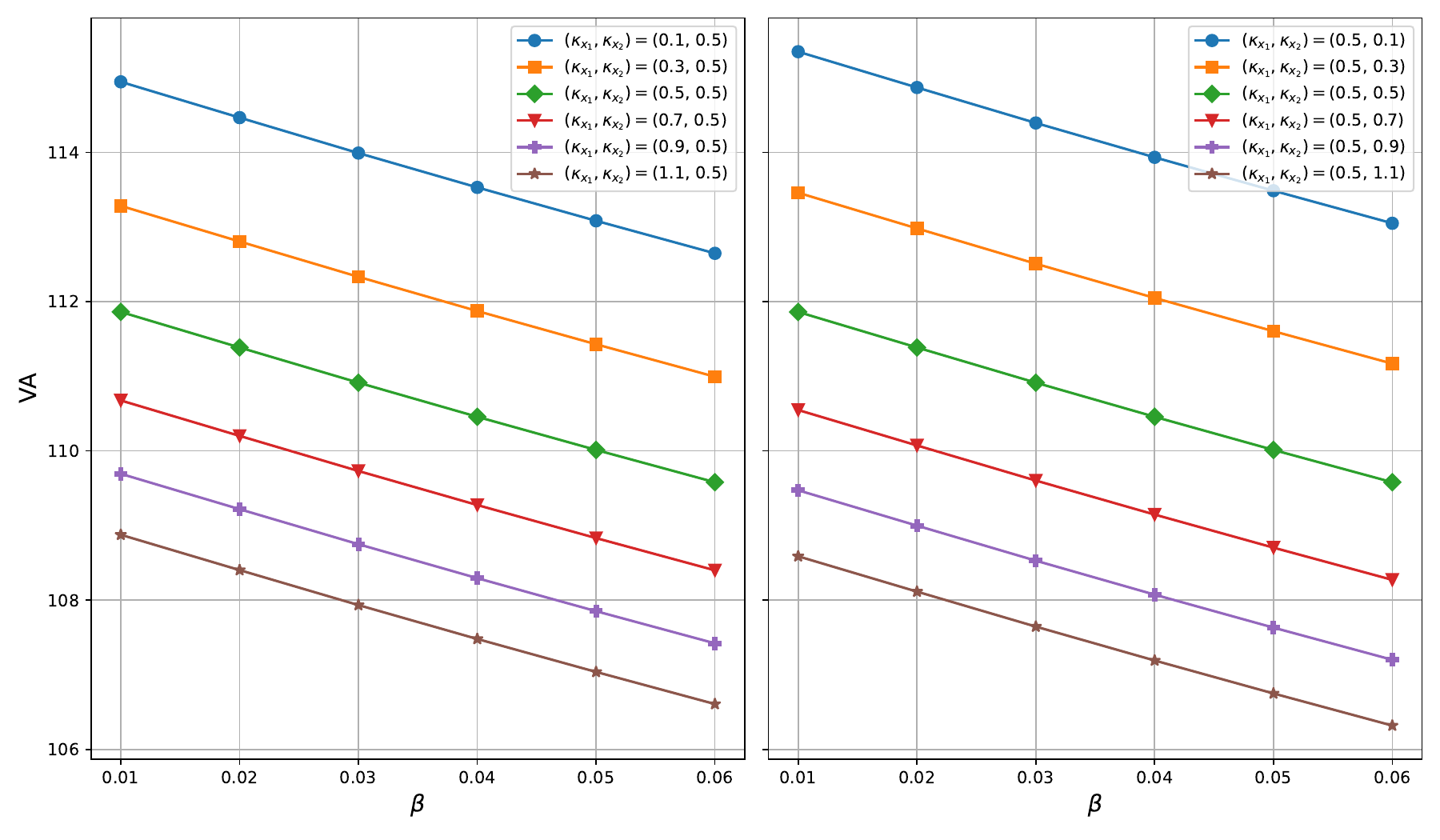}
	\caption{Sensitivity analysis : the mortality parameter $(\kappa_{x1},\kappa_{x2})$-VA. Maturity: T = 10 years.}
	\label{fig:Sensitivity_P_to_mort_kappa}
\end{figure}

Figure \ref{fig:Sensitivity_P_to_mort_kappa} shows the effect of the long-term dependence structure on the life policy's value. As $\kappa_q$ increases, the contract value decreases, which is consistent with pricing intuition. From its definition in \eqref{rq}, when $\kappa_q$ rises, the bereavement effect $r_q$ decreases. This reduction in $r_q$ implies a weaker impact of the broken heart syndrome on the surviving spouse's mortality, leading to a lower risk. Consequently, the policy price decreases to reflect this reduced risk

\appendix

\section{Some important  representations}\label{app:prelim}

We recall the following notation from Section \ref{secgenset}. \begin{eqnarray*}
	A(u,T)= \int_u^T \alpha(u,s)\mathrm{d}s, \ \ \ \Sigma_1(u,T)= \int_u^T \sigma_1(u,s)\mathrm{d}s,\ \ \
	\Sigma_2(u,T) = \int_u^T \beta(u,s)\mathrm{d}s. 
\end{eqnarray*}
We have the following  representation.
\begin{lemma}\label{lemm1a} Let $f(t,s)$ be the quantity defined in \eqref{forward rate}. For any $0\leq t\leq T$, we have \label{lem:HJM}
	$$
	-\int_{t}^T f(t,s)\mathrm{d}s = \int_0^{t}r(s)\mathrm{d}s -\int_0^T f(0,s)\mathrm{d}s - \int_0^{t}A(u,T)\mathrm{d}u + \int_0^{t}\Sigma_1(u,T)\mathrm{d}L^1_u - \int_0^{t}\Sigma_2(u,T)\mathrm{d}L^2_u.
	$$
\end{lemma}

\begin{proof}
	
	Applying the Fubini's theorem for stochastic integrals, we obtain
	\begin{align*}
	& -\int_{t}^T f(t,s)\mathrm{d}s \\
	 =& \int_0^{t}r(s)\mathrm{d}s -\int_0^{t} r(s)\mathrm{d}s - \int_{t}^T f(0,s)\mathrm{d}s - \int_{t}^T \int_0^{t}\alpha(u,s)\mathrm{d}u\, \mathrm{d}s\\
		& +\int_{t}^T\int_0^{t}\sigma_1(u,s)\mathrm{d}L^1_u\mathrm{d}s -\int_{t}^T\int_0^{t}\beta(u,s)\mathrm{d}L^2_u\mathrm{d}s \\
		=& \int_0^{t}r(s)\mathrm{d}s -\Big[\int_0^{t}f(0,s)\mathrm{d}s + \int_0^{t}(\int_0^s \alpha(u,s)\mathrm{d}u)\mathrm{d}s - \int_0^{t}(\int_0^s \sigma_1(u,s)\mathrm{d}L^1_u)\mathrm{d}s  \\
		&+ \int_0^{t}(\int_0^s \beta(u,s)\mathrm{d}L^2_u)\mathrm{d}s\Big] - \int_{t}^T f(0,s)\mathrm{d}s - \int_0^{t}\int_{t}^T \alpha(u,s)\mathrm{d}s\, \mathrm{d}u +\int_0^{t}\int_{t}^T\sigma_1(u,s) \mathrm{d}s\,\mathrm{d}L^1_u\\
		& - \int_0^{t}\int_{t}^T\beta(u,s) \mathrm{d}s\,\mathrm{d}L^2_u\\
		=& \int_0^{t}r(s)\mathrm{d}s  -\Big[\int_0^{t}f(0,s)\mathrm{d}s + \int_0^{t}\int_u^{t} \alpha(u,s)\mathrm{d}s\,\mathrm{d}u - \int_0^{t}\int_u^{t} \sigma_1(u,s)\mathrm{d}s\,\mathrm{d}L^1_u + \int_0^{t}\int_u^{t} \beta(u,s)\mathrm{d}s\,\mathrm{d}L^2_u\Big] \\
		& - \int_{t}^T f(0,s)\mathrm{d}s - \int_0^{t}\int_{t}^T \alpha(u,s)\mathrm{d}s\mathrm{d}u +\int_0^{t}\int_{t}^T\sigma_1(u,s) \mathrm{d}s\,\mathrm{d}L^1_u - \int_0^{t}\int_{t}^T\beta(u,s) \mathrm{d}s\mathrm{d}L^2_u\\
		=& \int_0^{t}r(s)\mathrm{d}s - \int_{0}^T f(0,s)\mathrm{d}s - \int_0^{t}\int_{u}^T \alpha(u,s)\mathrm{d}s\mathrm{d}u +\int_0^{t}\int_{u}^T\sigma_1(u,s) \mathrm{d} s\mathrm{d}L^1_u - \int_0^{t}\int_{u}^T\beta(u,s) \mathrm{d} s\mathrm{d}L^2_u.
	\end{align*}
The result follows from the definition of $A(u,T), \Sigma_1(u,T)$ and $\Sigma_2(u,T)$.

\end{proof}

Using the previous lemma, we get  the following representation of $D(t)$ defined in \eqref{Dt}.
\begin{align}
	D(t)=& \int_0^t r(s)\mathrm{d}s+\int_0^t \sigma_2(s)\mathrm{d}L^2_s -\omega(t)  -p(t) + \int_{t}^Tf(t,s)\mathrm{d}s - \delta T  \notag \\
	=&  -p(t) - \delta T+\int_0^t r(s)\mathrm{d}s+\int_0^t \sigma_2(s)\mathrm{d}L^2_s -\omega(t)- \int_0^{t}r(s)\mathrm{d}s +\int_0^T f(0,s)\mathrm{d}s  \notag\\
	&+ \int_0^{t}A(u,T)\mathrm{d}u - \int_0^{t}\Sigma_1(u,T)\mathrm{d}L^1_u + \int_0^{t}\Sigma_2(u,T)\mathrm{d}L^2_u \notag\\
	=& -p(t) - \delta T+\int_0^t \sigma_2(s)\mathrm{d}L^2_s -\omega(t)+\int_0^T f(0,s)\mathrm{d}s + \int_0^{t}A(u,T)\mathrm{d}u  \notag\\
	&- \int_0^{t}\Sigma_1(u,T)\mathrm{d}L^1_u+ \int_0^{t}\Sigma_2(u,T)\mathrm{d}L^2_u. \label{D:appendix}
\end{align}
By setting $t=T$ in the previous lemma, we get 
\begin{equation}
	0 = \int_0^{t}r(s)\mathrm{d}s -\int_0^t f(0,s)\mathrm{d}s - \int_0^{t}A(s,t)\mathrm{d}s + \int_0^{t}\Sigma_1(s,t)\mathrm{d}L^1_s - \int_0^{t}\Sigma_2(s,t)\mathrm{d}L^2_s, \label{interest}
\end{equation}
therefore, we deduce the following representation for the bank account (see (49) in \cite{Eberlein-R}),
\begin{equation}
	B(t) = \frac{1}{B(0,t)}\exp\Big( \int_0^t A(s,t)\mathrm{d}s - \int_0^t \Sigma_1(s,t)\mathrm{d}L^1_s +\int_0^t \Sigma_2(s,t)\mathrm{d}L^2_s \Big). \label{Bank ac}
\end{equation}

\end{document}